\RequirePackage{fix-cm}
\documentclass{svjour3-HAL}                     
\usepackage{verbatim}
\usepackage{stackengine}
\usepackage{amssymb}
\usepackage{amsmath}
\usepackage{graphicx}
\usepackage{comment}
\usepackage[noend]{algorithmic}

\usepackage{amsmath}
\usepackage{wasysym}
\usepackage{booktabs}

\usepackage[ruled]{algorithm}
\usepackage{algorithmic}

\usepackage{url}
\urldef{\mailsb}\path|nicolas.spyratos@lri.fr|

\usepackage{todonotes}

\begin{document}

\title{Tree Databases}



%
%

\author{Nicolas Spyratos}
\institute{LISN Laboratory - University Paris-Saclay, CNRS, INRIA, France\\
Affiliated Scienist, FORTH Institute of Computer Science,  Greece\\
\mailsb\\~\\
{\bf Acknowledgment:} Work conducted while the author was visiting with the HCI Group at FORTH Institute of Computer Science, Greece (https://www.ics.forth.gr/)}
\date{}

\maketitle

\begin{abstract}
We propose a novel database model whose basic structure is a labeled, directed tree. Intuitively, the root of the tree is seen as an object (or entity), the non-root nodes as attributes of the object and the semantics of each attribute is represented by the unique path leading from the root to the attribute. 
We define a tree database to be a set of such trees. The trees of the database can be combined to produce new trees using a set of operations on trees that we define in the paper. 
The query language of our model offers two types of queries, traversal queries and analytic queries. 
A query (whether traversal or analytic) is always defined over a tree, 
which is either a tree in the database or a tree derived from other trees using tree operations. The operations on trees and the query language are both defined using a simple functional algebra whose operations are: restriction of a function, composition of functions, pairing of functions and Cartesian product of sets.

\noindent A distinctive feature of our model is that traversal queries and analytic queries are both defined within the same formal framework; and in fact, traversal queries serve as the building blocks for analytic queries. This is in sharp contrast with the relational model, where analytic queries are defined outside the relational algebra, in the form of SQL Group-by queries. Therefore our model supports data access and data analysis within the same formal framework.  

\noindent We demonstrate the expressive power of our model by showing: (a) how our model can support inheritance in a seamless manner, (b) how one can define consistent relational databases on top of a tree database - with the tree database playing the role of an underlying semantic layer and (c) how a tree database can be used as a user-friendly interface for accessing and analyzing relational data.

\end{abstract}


\section{Introduction}

The basic idea underlying this work is that the data sets of an application and their relationships can be modeled as a labeled, directed tree, in which ach node stores a data set of the application and each edge stores a total function (i.e. a set of key-value pairs) from its source node to its target node. A set of such trees is what we call a {\em tree database} (or simply database, hereafter). Intuitively, the root of a tree in the database is seen as an object (or entity), the non-root nodes as attributes of the object and the semantics of each attribute is represented by the unique path leading from the root to the attribute. 


\noindent An example of a tree database is shown in Figure \ref{Fig-1}. It contains four trees, $\mathcal T_1, \mathcal T_2, \mathcal T_3, \mathcal T_4$, storing data concerning a distribution center (e.g. Walmart) which delivers products of various types in a number of branches. A delivery invoice contains the following data: the date of delivery, the branch in which the delivery took place, the type of product delivered (e.g. $Coca Light$) and the quantity (i.e. the number of units delivered of that type of product). There is a separate invoice for each type of product delivered, and the data on all invoices during the year are stored in $\mathcal T_1$ for yearly analyses and planning purposes. At any given moment during the year, the root $Inv$ of $\mathcal T_1$ stores the set of all invoice identifiers so far (e.g. a set of integers) and the non-root nodes store the data appearing on the invoices (e.g. the node `Date' stores the delivery dates appearing on the invoices, the node `Branch' stores all identifiers of branches to which deliveries were made, and so on). As for the edges of $\mathcal T_1$, each edge  $e: Inv\to X $ stores a function (i.e. a set of key-value pairs) associating each invoice number $i$ of $Inv$ with the corresponding data $e(i)$ stored in node $X$ (e.g. the edge $b: Inv \to Branch$ associates each invoice number $i$ of $Inv$ with the branch identifier $b(i)$ of $Branch$ appearing on invoice $i$). 
By the way, to simplify our discussions, we shall use the terms `edge' and `function' interchangeably whenever no ambiguity is possible; and similarly we shall use interchangeably the terms `node' and `set'. 

\noindent The remaining trees of our example database store auxiliary information needed for the analysis of the data stored in tree $\mathcal T_1$. Thus $\mathcal T_2$ stores the region in which each branch is located; $\mathcal T_3$ stores the supplier and the category of each product type; and $\mathcal T_4$ 
stores the unit price of a product as a function of its supplier and its category. 
\begin{figure}\label{Fig-1}
{
\begin{center}
\includegraphics[width=350px,keepaspectratio]{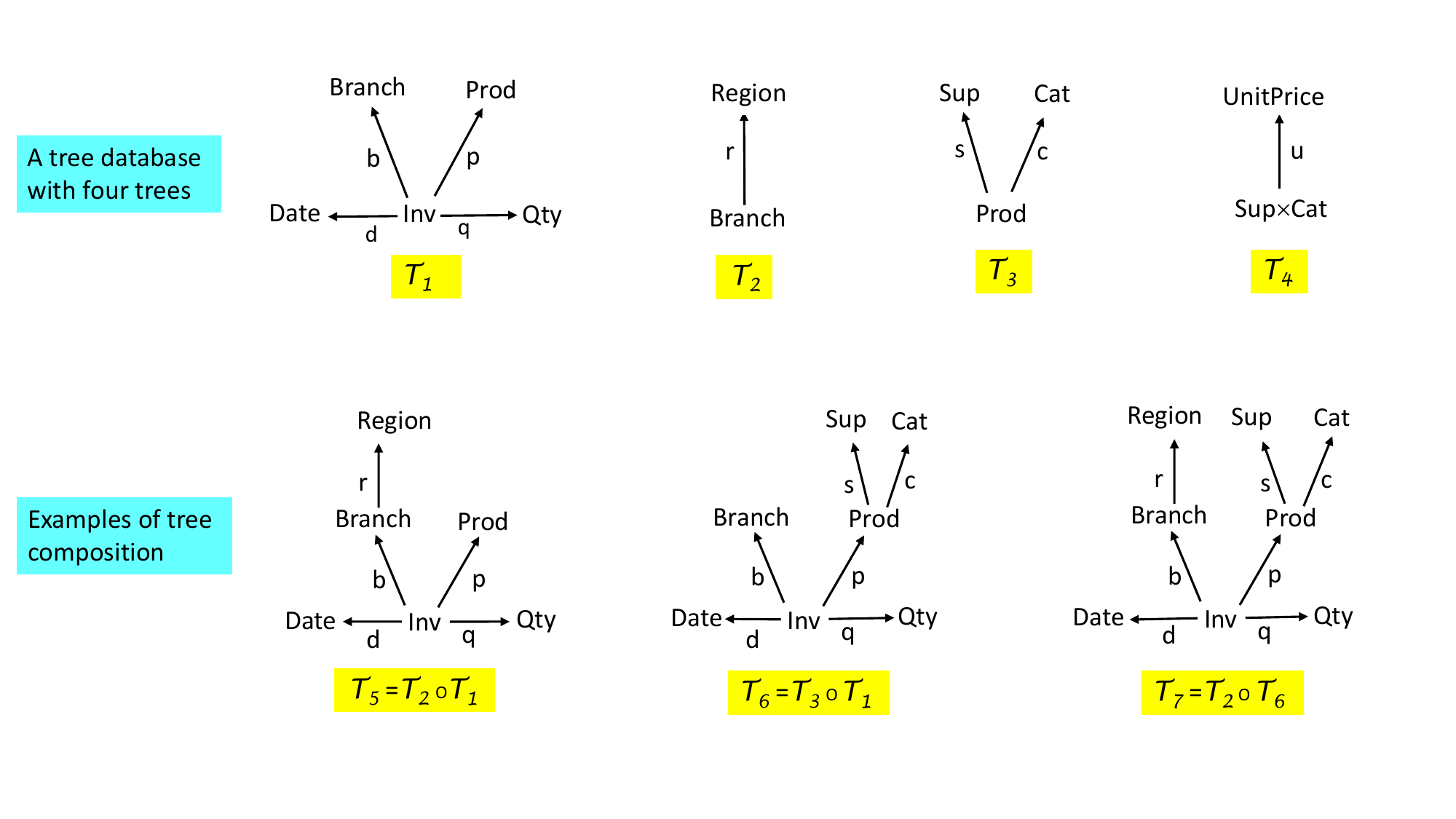}
\caption{A tree database and examples of tree composition} 
\end{center}
}
\end{figure} 
\smallskip 

\noindent Now, the trees of a database can be combined to produce new trees using a set of operations on trees that we shall define in the paper. Such operations on trees create new trees and therefore new ways of accessing information. An example of tree operation is `tree composition' defined as follows: a tree $\mathcal T$ is composable with tree $\mathcal T'$ if the root of $\mathcal T'$ is a node of $\mathcal T$ and, additionally, this is the only node shared by the two trees (this unique shared node is then called the `link' of the comopsition). If $\mathcal T$ is composable with tree $\mathcal T'$ then their composition is denoted by $\mathcal T'\circ \mathcal T$ and the result is a new `larger' tree. \smallskip

\noindent Figure \ref{Fig-1} shows three examples of tree composition. First, $\mathcal T_1$ is composable with $\mathcal T_2$ (with the node $Branch$ as the link) and the result of their composition is the tree $\mathcal T_5= \mathcal T_2\circ \mathcal T_1$ as shown in the figure; similarly, $\mathcal T_1$ is composable with $\mathcal T_3$ (with the node $Prod$ as the link) and the result of their composition is the tree $\mathcal T_6= \mathcal T_3\circ \mathcal T_1$; and finally, $\mathcal T_6$ is composable with $\mathcal T_2$ (with the node $Branch$ as the link) and the result of their composition is the tree $\mathcal T_7= \mathcal T_2\circ \mathcal T_6$. By the way, we shall use the tree $T_7$ as our running example in most of our discussions throughout the paper. \smallskip

\noindent Suppose now that $\mathcal T$ is a tree, which is either in the database or derived from those in the database as explained above. Then the question is: how can one extract information from $\mathcal T$? To this end, we need operations to combine nodes and edges in $\mathcal T$ so as to formulate queries. In our model we use one operation on sets, namely {\em Cartesian product}; and three operations on functions, namely {\it restriction} of a function to a subset of its domain of definition, {\it composition} of two functions, and {\it pairing} of two functions with the same source. We shall refer to these operations, collectively, as the {\em functional algebra}. Note that these are well known, elementary operations except possibly for pairing which is defined as follows: given two functions $f: X \to Y$ and $g: X \to Z$ with $X$ as their common source, we call {\em pairing} of $f$ and $g$, denoted as $f \wedge g$, the function defined by:
\begin{center}
 $f \wedge g: X \to Y \times Z$ such that $(f \wedge g)(x)= (f(x), g(x))$ for all $x$ in $X$    
\end{center}
Note that pairing works as a tuple constructor. Indeed, if we view the elements of $X$ as identifiers (or `keys'), then for each $x$ in $X$ the pairing constructs a tuple $(x, f(x), g(x))$ consisting of $x$ and the images of $x$ under the input functions. Clearly, the definition of pairing can be extended to more than two functions with the same source in a straightforward manner. Also note that function composition and function pairing express the only two ways edges combine in a graph (i.e. either forming a path or a fork).  \smallskip

\noindent Our model offers a query language comprising two types of queries, {\em traversal queries} and {\em analytic queries}. They are both defined using the functional algebra defined above. Given a tree $\mathcal T$, a {\em path expression} over $\mathcal T$ is just the composition of the edges of a path of $\mathcal T$; and a traversal query over $\mathcal T$ is either a single path expression  or the pairing of two or more path expressions having the same source. Referring to Figure \ref{Fig-1}, here are two examples of path expressions (and therefore of traversal queries) over the tree $T_7$: \smallskip

\noindent $PE_1: r \circ b$, returning, for each $i\in Inv$, the region associated with $i$ 

\noindent $PE_2: s\circ p$, returning, for each $i\in Inv$, the  product associated with $i$ \smallskip

\noindent By pairing these two path expressions we obtain the following traversal query:\smallskip

\noindent $Q= (r \circ b)\wedge (s\circ p)$, returning for each $i\in Inv$, the  region-product pair associated with $i$ \smallskip

\noindent Note that the answer of a path expression is a function from the source of the path to its target; and that the answer of a traversal query is a function from the common source of its path expressions to the Cartesian product of the targets of its path expressions. In order to simplify our discussions, we will often confuse a traversal query with its answer, which is a function. For example we will say `the function $r \circ b$' to mean the answer to the traversal query  $r \circ b$. \smallskip

\noindent Similarly, an {\em analytic expression} over a tree $\mathcal T$ is a triple of the form $\langle g, m, op \rangle$, where $g$ and $m$ are traversal queries over $\mathcal T$ with the same source (serving for grouping and measuring, respectively) and $op$ is an aggregate operation applicable on the target of $m$; and an analytic query is either a single analytic expression  or the pairing of two or more analytic expressions having the same source. Referring to Figure \ref{Fig-1}, here are two  examples of analytic expressions over the tree $T_7$: \smallskip

\noindent$ AE_1: (b, q, sum)$ is an analytic expression returning the total quantity delivered by $Branch$ 

\noindent $AE_2: (b, q, avg)$ is an analytic expression returning the average quantity delivered by $Branch$ \smallskip

\noindent To compute the answer to $AE_1$ we proceed as follows: first, given a branch $x$, we group together all invoices that refer $x$ (actually, this is the set $S= b^{-1}(x)$); then we apply the function $q$ to every invoice $i$ in $S$ to find the quantity corresponding to $i$; and finally, we sum up the quantities for all $i\in S$ to find the total quantity for branch $x$ (and similarly for computing the answer of $AE_2$). Note that the answer to an analytic expression is a function from the target of the grouping query $b$ to the target of the aggregate operation (and as this target is a node not appearing in the tree, the user is asked to name it (e.g. $BranchTot$ in our example). The answer to $AE_2$ is computed in a similar manner returning the average quantity delivered by branch, say $BranchAvg$. Next, by pairing these two analytic expressions we obtain the following analytic query:\smallskip

\noindent $Q'= (b, q, sum)\wedge (b, q, avg)$, 

\noindent returning, for each $i\in Inv$, the pair $(BranchTot, BranchAvg)$ associated with $i$ \smallskip

\noindent Note that the answer of a an analytic expression $(g, m, op)$ is a function from the target of $g$ to the target of $op$; and the answer of an analytic query is a function from the common source of its analytic expressions to the Cartesian product of the targets of its analytic expressions.\smallskip

\noindent Roughly speaking, traversal queries parallel relational algebra queries, whereas analytic queries parallel SQL Group-by queries. A distinctive feature of our model is that traversal queries and analytic queries are both defined within the same formal framework; and in fact, traversal queries serve as the building blocks for analytic queries. This is in sharp contrast to the relational model, where analytic queries are defined outside the relational algebra, in the form of SQL Group-by queries. Therefore our model supports data access {\em and} data analysis within the {\em same} formal framework.\smallskip

\noindent We demonstrate the expressive power of our model by showing: (a) how our model can support inheritance in a seamless manner, (b) how one can define consistent relational databases on top of a tree database - with the tree database playing the role of an underlying semantic layer and (c) how a tree database can be used as a user-friendly interface for accessing and analyzing relational data. \smallskip

\noindent A fundamental aspect of our model is that data is stored in the form of trees that is data is not stored sequentially (i.e., not in a linear order). Instead, data is organized across multiple levels, forming a hierarchical structure. 
Graph data models in general have gained considerable attention in the last few decades since they use nodes, relationships between nodes and key-value properties instead of tables to represent information \cite{anuyah2024understandinggraphdatabasescomprehensive}\cite{DBLP:reference/bdt/GutierrezHW19}. Therefore they are typically substantially faster for associative data sets. 
While other database models compute relationships expensively at query time, a graph database stores connections as first-class citizens, readily available for any `join-like' navigation operation. Any purposeful query over a graph database can be easily answered, as data is easily accessed using traversals. A traversal is how you query a graph, navigating from starting nodes to related nodes according to an algorithm \cite{DBLP:journals/corr/AnglesABHRV16}. Graph-based models present a high degree of flexibility and are able to capture complex non-linear relationships. In a graph data model relationships are of equal importance to the data itself. This means you are not required to infer connections between entities using special properties such as foreign keys or out-of-band processing like map-reduce. \smallskip

\noindent In contrast, relational databases store highly-structured data in tables with predetermined columns and rows of specific types of information. Due to the rigidity of their organization, relational databases require developers and applications to strictly structure the data used in their applications. In relational databases, references to other rows and tables are indicated by referring to primary key attributes via foreign key columns; joins are computed at query time by matching primary and foreign keys of all rows in the connected tables. These operations are compute-heavy and memory-intensive, and have an exponential cost; and when many-to-many relationships occur in the model, you must introduce a join table (or associative entity table) that holds foreign keys of the participating tables, further increasing join operation costs.

 
\noindent Graph databases have become a cornerstone of modern data management, enabling complex queries, pattern discovery, and knowledge inference in highly connected datasets \cite{DBLP:reference/bdt/GutierrezHW19}. There is a substantial body of literature around the use of graphs in computer science, including several tools to support graph management by means of digital technology, as well as data models and query languages based on graphs. The reader is referred to \cite{anuyah2024understandinggraphdatabasescomprehensive}\cite{DBLP:conf/sigmod/ArenasGS21}\cite{DBLP:conf/aib/Hogan22}\cite{DBLP:journals/csur/HoganBCdMGKGNNN21}\cite{0c1f962ff899492c9ea15bd38a0ffb65} for a comprehensive analysis of graph-based approaches to data and knowledge management, emphasizing the relation between graph databases and knowledge graphs, and including an extensive bibliography. A rather detailed discussion on the power and limitations of graph databases can be found in \cite{DBLP:conf/cisim/Pokorny15}. Moreover, several graph database systems have appeared in recent years such as Neo4j Graph Database, ArangoDB, Amazon Neptune, Dgraph and a host of others\footnote{https://www.g2.com/categories/graph-databases (retrieved on March 4, 2026)}. A comparison of the main graph database models can be found in \cite{DBLP:conf/data/FernandesB18}.

\section{The Formal Definition of a Tree Database}\label{sec:Formal}

In this section we give formal definitions of the basic concepts seen informally in the introduction, namely `tree database' and `functional algebra'; and we also introduce the important concept of integrity constraint over a tree database. 



\subsection{Tree Database}

To define formally the concept of tree database we need some auxiliary concepts and notation. \smallskip

\noindent First, we call {\em universe} any finite abstract set over which we define `nodes'. A {\em node} over $U$ is either an element of $U$ ({\em simple node}) or the Cartesian product of two or more elements of $U$ ({\em composite node} or {\em product node}). 
We assume that each simple node $A$ over $U$ is associated with a fixed set of values, called the {\em domain} of $A$ and denoted by $dom(A)$. The domain of a node may be a finite or an infinite set of values (e.g. the set of the two boolean values or the set of integers can be node domains).  
As for a product node, its domain is defined to be the product of the domains of its factors; for example, for $A, B$ in $U$, $dom(A \times B)= dom(A) \times dom(B)$. 

\smallskip 
\noindent Second, throughout the paper, by `tree over $U$' we mean a finite, directed and labeled graph whose nodes are nodes over $U$ such that: (a) there is a distinguished node with no incoming edge called the {\em root}, (b) there is exactly one path from the root to every non-root node and (c) all edges and all nodes have distinct labels.  Moreover, we call {\em leaf} any node with no outgoing edge and we assume that the set of successors of every node is unordered. 

\smallskip 
\noindent Third, we shall call {\em star tree} any tree in which all leaves are (immediate) successors of the root.
\begin{definition}[Tree Database]
Let $U$ be a universe. A {\em tree database} over $U$ is a finite set $\mathcal D= \{\mathcal T_1, \ldots, \mathcal T_n\}$ of trees over $U$ such that: 
    \begin{enumerate}
    \item In each tree, the root can be a simple or a composite node over $U$, while each non-root node must be a simple node. 
    \item Each simple node over $U$ must appear in at least one tree (i.e. no superfluous nodes in $U$).
    \item Each simple node $X$ over $U$ can appear at most once in a given tree $\mathcal T_i$  (but $X$ may appear in two or more different trees). 
    \item Every tree node $X$ is associated with a special edge $\iota_X: X \to X$, called the {\em identity edge} of $X$.
   \end{enumerate} 
\end{definition}

\noindent Motivation for the above definition of a tree database comes from the fact that an object (or entity) and its attributes can be seen as a tree in which the root node represents the object; the non root nodes represent the attributes of the object; and thus a tree database is seen as a set of objects. Moreover, an object may have a complex structure (hence a root is allowed to be a product node), whereas an attribute is seen as having no structure (hence the requirement that a non root node be a simple node). \smallskip

\noindent Note that, although a node can appear at most once in a given tree (item 3 of the above definition), a node may appear in two or more different trees that is two different objects may share the same attribute (e.g. an employee and a manager share the attribute `salary'). In other words, all nodes of $U$ must appear in the database (item 2) but each node can appear at most once in a given tree (either as a simple node or as a component of a product node). Also note that the above definition allows an edge to appear in two or more different trees (with the same or different edge labels) provided that its source and target nodes satisfy the constraints expressed by items 2 and 3. For example, consider the two trees: \smallskip

$\mathcal T: X\xrightarrow{f} Y\xrightarrow{g} Z$ ~~~and~~~ $\mathcal T': Y\xrightarrow{g} Z\xrightarrow{h} W$ \smallskip

\noindent Although the edge $Y\xrightarrow{g} Z$ is shared by $\mathcal T$ and $\mathcal T'$, the two trees can be in the same database.  Similarly, consider the two (one-edge) trees: \smallskip 

$\mathcal T: Emp\xrightarrow{works-for} Dep$ ~~~ and ~~~ $\mathcal T':Emp \xrightarrow{manages} Dep$ \smallskip

\noindent Although they have the same source and the same target (but different labels), they can be in the same database.  In both the above examples, what differentiates the nodes or edges shared by two or more  trees in a database are their instances - a notion that we will discuss shortly. \smallskip

\noindent The significance of the edge $\iota_X$ will become evident when discussing the query language of our model (sections \ref{sec:TQ} and \ref{sec:AQ}).\\

\noindent Now, as we explained informally in the introduction, the nodes and edges of a database tree represent the datasets of an application and their relationships; and these datasets and relationships may change with time. The contents of a database tree $\mathcal T$, at any given time $t$, is what we call the {\em instance} of $\mathcal T$ at time $t$; and the set of all tree instances at time $t$ is what we call the {\em database instance} at time $t$.


\begin{definition}[Tree Instance]\label{TrIn}
Let $U$ be a universe, $\mathcal D= \{\mathcal T_1, \dots, \mathcal T_n\}$ a database over $U$ and  $\mathcal T$ a tree in $\mathcal D$. An {\em instance} of $\mathcal T$ is a function $\delta_\mathcal T$ that associates the nodes and edges of $\mathcal T$ with values such that:
\begin{itemize}
    \item for each node $N$ of $\mathcal T$, $\delta_\mathcal T(N)$ is a finite nonempty subset of $dom(N)$ 
    \item for each edge $f: X \to Y$ of $\mathcal T$, $\delta_\mathcal T(f)$ is a  {\em total function} from $\delta_\mathcal T(X)$ to $\delta_\mathcal T(Y)$
    \item for each node $N$ of $\mathcal T$, $\delta_\mathcal T(\iota_N)$ is the identity function on $\delta_\mathcal T(N)$

\end{itemize}
Moreover, an {\em instance} of $\mathcal D$ at time $t$ is defined to be the set $\delta= \{\delta_{T_1}, \ldots, \delta_{T_n}\}$ of instances of all trees of $\mathcal D$ at time $t$.
\end{definition}

\noindent An important remark is in order here regarding the above definition. Although the domain of a node can be an infinite set, the instance of a node (whether simple or composite) is always a {\em finite} nonempty set of values from its domain; and as a consequence, the instance of an edge $f: X \to Y$ is always a finite nonempty function. Moreover, as we assume that $\delta_\mathcal T(f)$ is a {\em total} function, this means that $\delta_\mathcal T(f)$ is defined on every value of $\delta(X)$, which means that we assume {\em no nulls}. \smallskip
\noindent Another important remark is that, as mentioned earlier, a node $X$ may be shared by two or more different trees of the database, say $\mathcal T$ and $\mathcal T'$. In this case, $X$ may be associated with different instances in the two trees that is we may have $\delta_\mathcal T(X) \neq \delta_{\mathcal T'}(X)$.  For example, in Figure~\ref{Fig-1}, the node $Branch$ is shared by the trees $\mathcal T_1$ and $\mathcal T_2$ and we have, in general, $\delta_{\mathcal T_1}(Branch) \neq \delta_{\mathcal T_2}(Branch)$. This is natural as $\delta_{\mathcal T_1}(Branch)$ is the set of all branches to which products have been delivered, whereas $\delta_{\mathcal T_2}(Branch)$ is the set of all branches of the distribution company (i.e. including those branches to which no deliveries have been made).\smallskip


\noindent In the rest of the paper, in order to simplify our discussions, we shall not make the distinction between a node and its instance or between an edge and its instance - unless necessary. Indeed, more often than not, the intended meaning of terms will be evident from the way these terms are used. For example, if $e$ is an edge of a tree $\mathcal T$ and we write $e^{-1}$ this will clearly mean $\delta_\mathcal T(e)^{-1}$; and similarly, if $f: X \to Y$ and $g: Y \to Z$ are two edges and we write $g \circ f$ this will clearly mean the composition $\delta_\mathcal T(g) \circ \delta_\mathcal T(f)$; or if $P: X \xrightarrow{f_1} Y \xrightarrow{f_2}Z$ is a path and we write $P(x)$, where $x$ in $X$, this will clearly mean $(\delta_\mathcal T(f_2)\circ \delta_\mathcal T(f_1))(x))$; and similarly for a path of $n$ edges $f_1, f_2, \ldots, f_n, n\ge 1$.\smallskip

\subsection{The functional algebra}\label{subsec:FA} 
In order to access the data stored in a tree database users need operations to combine nodes and edges in a tree so as to formulate queries. In our model, as mentioned in the introduction, we use four well known, elementary operations that we call, collectively, the {\em functional algebra}. 

\begin{definition}[Functional algebra]\label{FA}
Given a tree $\mathcal T$, the functional algebra of $\mathcal T$ consists of the following operations: 
\begin{itemize}
    \item {\em Cartesian product} of sets (and its accompanying projection functions).
    \item {\em Restriction}: Given a function $f:X\to Y$ and a subset $S$ of $X$, the restriction of $f$ to $S$, denoted by $f/S$, is the function $f/S: S \to Y$ defined by: $(f/S)(x)= f(x)$, for all $x$ in $S$.
    \item {\em Pairing}: Given two functions $f:X \to Y$ and $g:X \to Z$ with common source, their pairing, denoted by $f\wedge g$, is the function $f\wedge g:X \to Y \times Z$ such that: $(f\wedge g)(x)= (f(x), g(x))$ for all $x$ in $X$.
    \item {\em Composition}: Given two functions $f:X \to Y$ and $g:Y \to Z$, their composition, denoted by $g \circ f$ is the function $g \circ f:X \to Z$ defined by: $(g \circ f)(x)= g(f(x))$ for all $x$ in $X$.
\end{itemize} 
\end{definition} 

\noindent Note that the operations of composition and pairing can be defined for more than two functions in the obvious way. \smallskip

\noindent An important remark concerning the functional algebra is that its four operations are strongly connected to each other as stated in the following lemma. Its proof is a direct consequence from the definitions of the operations.

\begin{lemma}\label{BasicLemma} Let $X, Y, Z$ be three sets, and let $f: X \to Y$ and $g: X \to Z$ be two functions with common source. Then the following hold: 

\noindent $\pi_Y \circ (f \wedge g)= f$ ~and~ $\pi_Z \circ (f \wedge g)= g$

\noindent (where $\pi_Y$ and $\pi_Z$ denote the projections of $Y\times Z$ over $Y$ and $Z$, respectively).
\end{lemma}

\smallskip

\noindent There is an interesting `derived' operation, called `product of functions', defined as follows:


\begin{definition}[Product of functions]
Let $f: X \to Y$ and $g: X' \to Y'$ be two functions. The {\em product} of $f$ and $g$ is the function $f \times g: X \times X' \to Y \times Y'$ defined by: $(f \times g)(x, x')= (f(x), g(x'))$ for all $(x, x')$ in $X \times X'$.
\end{definition}

\noindent Clearly, the above definition of product can be extended to more than two functions in a straightforward manner. The following lemma states how the product can be derived from operations of the functional algebra, namely using projection, composition and pairing. Its proof follows immediately from the definitions. 

\begin{lemma}\label{Prod}
Let $f: X \to Y$ and $g: X' \to Y'$ be two functions. Then we have: $f \times g= (f \circ \pi_X) \wedge (g \circ \pi_{X'})$.

\end{lemma}

\noindent We now introduce the basic concept of `expression' over a database tree and its evaluation in a database instance. 

\begin{definition}[Expression over a Tree] \label{Expr}
    Let $\mathcal D$ be a database and let $\mathcal T$ be a tree in $\mathcal D$. An {\em expression} over $\mathcal T$ is either an edge of $\mathcal T$ or a well formed expression $E$ whose operands are edges of $\mathcal T$ and whose operations are among those of the functional algebra. Moreover, given an instance $\delta_\mathcal T$, the {\em value} of $E$ in $\delta_\mathcal T$, denoted as $Val_E(\delta_\mathcal T)$, is obtained by (a) replacing the nodes and edges of $E$ with the values assigned to them by $\delta_\mathcal T$ and (b) performing the operations.
\end{definition}

\noindent Note that this definition is in the spirit of a similar definition in the relational model, namely the definition of relational expression over a database schema \cite{Ullman}. 
Indeed, if $\mathcal S$ is a database schema and $\mathcal D$ a relational database over $\mathcal S$ then a relational expression over $\mathcal S$ is a well formed expression $E$ whose operands are relation schemas of $\mathcal S$ and whose operations are among those of the relational algebra. Moreover, the value of $E$ in $\mathcal D$ is obtained by replacing the relation schemas of $E$ with the relations assigned to them by $\mathcal D$ and performing the operations. \smallskip

 \noindent It is important to note that every expression $E$ has a source and a target that can be defined recursively based on the sources and targets of the edges appearing in $E$. For example, referring to tree $\mathcal T_7$ of Figure \ref{Fig-1}, if $E_1= r \circ b$ then $source(E_1)= Inv$ and $target(E_1)= Region$; and similarly, if $E_2= (r \circ b) \wedge p$  then $source(E_2)=Inv$ and $target(E_2)= Region \times Product$. Clearly, the value of an expression is always a function from the source of the expression to its target. \smallskip

 \noindent A particular kind of expression over a tree, called `path expression' will be frequently used in the rest of this paper. Intuitively, a path expression is the composition of edges along a path. 
\begin{definition}[Path expression]
 Let  $\mathcal T$ be a database tree. A {\em path expression} over $\mathcal T$ is a well formed expression $PE$ whose operands are the edges of a path of $\mathcal T$ and whose only operation is composition. Moreover, given an instance $\delta_\mathcal T$, the value of $PE$ is defined to be the composition of the functions assigned to the edges of $PE$ by $\delta_\mathcal T$ (therefore, a path expression with source $S$ and target $A$ evaluates to a function from $S$ to $A$). 
\end{definition}
\noindent Note that, as there is at most one path between any two nodes of a tree, every path is associated to one and only one path expression, namely the composition of its edges; and conversely, every path expression is associated to one and only one path (namely, the sequence of its operands in reverse order). \smallskip

\noindent We end this section with an interesting observation relating the functional algebra with functional dependency theory, and in particular with Armstrong's axioms for functional dependencies in relational databases~\cite{DBLP:conf/ifip/Armstrong74}\cite{DBLP:books/cs/Maier83}. Let $\mathcal T$ be a database tree, $\delta_\mathcal T$ an instance of $\mathcal T$ and $f: X \to Y$ a function, where $X, Y$ are nodes of $\mathcal T$. We shall say that $\delta_\mathcal T$ {\em implies} $f$ if $f$ is the value of an expression over $\mathcal T$, for all $\delta_\mathcal T$. Then the question is: what is the set of all functions that are `implied' by $\delta_\mathcal T$? The following lemma gives a first answer. 
\begin{lemma}\label{AA}
 Let $\delta_\mathcal T$ be an instance of a database tree $\mathcal T$.  Then the following hold:
 \begin{itemize}
     \item If $X$ is a simple node then  $\delta_\mathcal T$ implies the function $\iota_X$; and if $X$ is a product node and $Y$ a sub-product of $X$ then $\delta_\mathcal T$ implies the function $\pi_Y: X \to Y$.
     \item If two functions $f:X \to Y$ and $g:Y \to Z$ are implied by $\delta_\mathcal T$ then $g\circ f:X \to Z$ is also implied by $\delta_\mathcal T$.
     \item If the function $f:X \to Y$ is implied by $\delta_\mathcal T$ and $Z$ is a node of $\mathcal T$ then $f \times \iota_Z:X \times Z \to Y \times Z$ is also implied by $\delta_\mathcal T$.
 \end{itemize}
 \end{lemma}
 
\noindent{\em Proof} The first implication follows from the definition of Cartesian product of sets; the second follows by applying composition of functions; and the third follows by taking the product of $f$ with the identity function of $Z$.  \smallskip

\noindent We recall here that the properties expressed in the above lemma correspond to Armstrong's axioms, or implication rules for functional dependencies in relational databases (namely, reflexivity, transitivity and augmentation, respectively)~\cite{DBLP:conf/ifip/Armstrong74}\cite{DBLP:books/cs/Maier83}). However, the above lemma shows only the `soundness' of these implication rules, therefore the accompanying question is the following: is there a `derivation' system showing completeness? Answering this kind of questions lies outside the scope of the present paper and it is part of our current work. 

\smallskip
\noindent We end this section by summarizing a few basic properties of the operations of the functional algebra (whose proofs follow easily from the definitions):

\begin{enumerate}
    \item Strictly speaking, the Cartesian product is neither associative nor commutative. However, there are always natural bijections between the various forms of a Cartesian product, for example, between $A \times B \times C$, $(A \times B) \times C$ and $A \times (B \times C)$, or between $A \times B$ and $B \times A$, and so on. Therefore, in our setting, we shall assume that the Cartesian product is both associative and commutative. As a consequence of this assumption, pairing is also associative and commutative. 
    \item Composition is an associative operation, therefore we can group together its component functions as it is more convenient for our purposes. For example, $(h \circ s) \circ p= h \circ (s \circ p)$. 
    \item Composition distributes over pairing. For example referring to tree $\mathcal T_7$ of Figure~\ref{Fig-1} we have $(s\wedge c)\circ p= (s\circ p)\wedge (c\circ p)$.
    \item For any functions $f:X \to Y$ and $g:Y \to Z$ we have: 
    
    if $S \subseteq X$ then $g \circ (f/S)= (g \circ f)/S$ 
    \item For any functions $f:X \to Y$ and $g:X \to Z$ we have: 
    
    if $S\subseteq X$ then $(f \wedge g)/S= (f/S) \wedge (g/S)$ 
\end{enumerate}

\noindent As a last remark, the functional algebra introduced in this section constitutes the foundation of our model and the basic tool for defining the tree operations and the query language that we shall see in the following sections.

\subsection{Integrity constraints}\label{IC} 

Integrity constraints are properties that the database must satisfy so that the quality of information is maintained. Examples of such constraints from relational databases are domain constraints, key constraints, or referential constraints \cite{Ullman}. Integrity constraints ensure that changes in the database are performed in such a way that data integrity is not affected. \smallskip

\noindent In our model we distinguish two kinds of constraints, namely those imposed by the definition of database over a universe $U$, and those that may be imposed by the application. There are two constraints imposed by our definition of database, namely (a) every instance of a database tree $\mathcal T$ must associate each node $X$ with a {\em finite} subset $\delta_\mathcal T(X)$ of $dom(X)$ and (b) each edge $f: X \to Y$ with a {\em total} function  $\delta_\mathcal T(f):\delta_\mathcal T(X)\to \delta_\mathcal T(Y)$ (meaning that no nulls are allowed in the database). 

\noindent Incidentally, the second constraint implies a referential constraint similar to referential constraints in the relational model. Indeed as each edge $e: X \to Y$ of a tree is interpreted as a total function, every object of $\delta_\mathcal T(X)$ must `refer' to an object of $\delta_\mathcal T(Y)$ (i.e. no `dangling pointers').\smallskip

\noindent As for constraints that may be imposed on a database by the application, we consider two kinds in this paper: (a) an inter-tree constraint that we call {\em inclusion constraint} and (b) an intra-tree constraint that we call {\em refinement constraint} or {\em dependency constraint}. An inclusion constraint applies to two trees, $\mathcal T$ and $\mathcal T'$ that share a node, say $X$, and expresses that we must have that $\delta_\mathcal T (X)\subseteq \delta_{\mathcal T'}(X)$, in every database instance. We denote this by $\delta_\mathcal T (X)\sqsubseteq \delta_{\mathcal T'}(X)$. For example, in the database of Figure~\ref{Fig-1}, where the trees $\mathcal T_1$ and $\mathcal T_2$ share the node $Prod$ we may impose the inclusion constraint $\delta_{\mathcal T_1}(Prod)\sqsubseteq \delta_{\mathcal T_2}(Prod)$ expressing that the set of products delivered are among those available at the company. 
 \smallskip
\noindent A  refinement constraint on the other hand states that, given a tree $\mathcal T$, an expression $E$ must contain {\em finer information} than an expression $E'$ with the same source as $E$. It is an information-theoretic constraint referring to the information content of expressions having the same source. To understand the definition of this constraint, recall that the value (instance) of an expression $E$ over a tree is a total function and therefore its information content can be defined as the partition $p_E$ that this function induces on its domain of definition. It follows that we can compare two expressions having the same source using the ordering of the partitions they induce on their common source. 

\begin{definition}[Refinement constraint]\label{RC}
Let $\mathcal T$ be a tree and let $E, E'$ be expressions over $\mathcal T$ having the same source. Then $E$ is said to be {\em finer} than or equal to $E'$, denoted by $E \leq E'$, if $p_E \leq p_{E'}$. Moreover, if $E \leq E'$ then we say that $E$ {\em determines} $E'$, denoted by $E \to E'$.  
\end{definition}

\noindent The following lemma states an interesting property of refinement constraints also showing their relationship to functional dependencies of the relational model. 

\begin{lemma}\label{EdgeDep}
 Let $\mathcal T$ be a tree and let  $f: X \to Y$ and $g: X \to Z$  be two edges of $\mathcal T$ having the same source.   Then $f \leq g$ if and only if there is a function $h: Y \to Z$ such that $h \circ f = g$. Moreover, the function $h$ is unique.
\end{lemma}
\proof  Suppose that $p_f \leq p_g$, and define $h(y)= g(f^{-1}(y))$ for every $y$ in the range of $f$ (i.e. in the set of values of $f$). Then $h$ is a well defined function from $Y$ to $Z$. Indeed, as $p_f \leq p_g$, every block of $p_f$ is contained in a block of $p_g$ and therefore $g(f^{-1}(y))$ is a single value in the range of $g$. It follows that $h:range(f) \to Z$ is a well defined function and that $h \circ f = g$. 

\noindent In the opposite direction, suppose that $h \circ f = g$. It follows that $p_{h \circ f} = p_g$ and as $p_f \leq p_{h \circ f}$ we have that $p_f \leq p_g$. 

\noindent Now, suppose that there is a function $h':Y \to Z$ such that $h \circ f = g'$. If $h' \ne h$ then there is $y \in range(f)$ such that $h'(y) \ne h(y)$. It follows that there is  $x \in X$ such that $f(x)= y$, therefore $h'(f(x)) \ne h(f(x))$, that is $g(x) \ne g(x)$ which is impossible as $g$ is a function.
\noindent Clearly the proof remains the same if $f, g$ and $h$ in the above lemma are values of expressions over $\mathcal T$ (since the value of an expression over $\mathcal T$ is a function).
\hfill$\Box$

\smallskip
To see an example of refinement constraint refer to Figure~\ref{Fig-1} and consider the edges $c$ and $s$ of tree $\mathcal T_3$ that associates each product with a category and a supplier. If we declare the refinement constraint $c \leq s$ this means that each product category must be provided by one and only one supplier. In other words, a refinement constraint works in much the same way as a functional dependency in the relational model. Let's see a simple example illustrating this point. Consider the following tree relating employees, departments and managers in some enterprise: 

$\mathcal T: Mgr \xleftarrow{g} Emp \xrightarrow{f} Dep$

\noindent If we impose the refinement constraint $f \le g$ then the above lemma implies that there must be a function $h: Dep \to Mgr$ which means that a department must have one and only one manager.  \smallskip

\noindent We shall not discuss constraints any further in this paper. Our intention in this section was to give a brief account of the kind of constraints that we can use in our model. A detailed discussion of constraints and their inference mechanism lies outside the scope of the present paper and it is part of our current work. \\

 \noindent We have seen so far the basic concepts of our model, namely the concepts of tree database, functional algebra and the basic constraints that can be imposed on a tree database. In the following two sections we define the two kinds of queries offered by the query language of our model, namely traversal queries and analytic queries. \smallskip
 
 \noindent We end this section by reminding the reader that, in order to simplify our discussions, we often confuse syntax and semantics, when no ambiguity is possible. Here are some examples, referring to tree $\mathcal T_7$ of Figure~\ref{Fig-1}: if we say `the function $b$' this means the function $\delta_{\mathcal T_7}(b)$; if we say `the function $r \circ b$' this will mean the function $\delta_{\mathcal T_7}(r) \circ \delta_{\mathcal T_7}(b)$; if we write $b^{-1}$ this will mean $(\delta_{\mathcal T_7}(b))^{-1}$ and so on.

\section{Operations on Trees}
In this section, we introduce operations for defining new trees from old. To do so, we `upgrade' the operations of the functional algebra to similar operations on trees. 
Before defining our operations on trees, we recall a few facts about graphs which also apply to trees. In particular, we recall the definition of a directed labeled graph as a set of triples. 

\noindent Seen as syntactic objects, the edges of a directed labeled graph are triples of the form $\langle source, label, target\rangle$, therefore two edges are different if they differ in at least one component of this triple. This implies, in particular, that two edges can have the same label if they have different sources and/or different targets. Moreover, two different edges can have the same source and the same target provided that they have different labels; we call such edges {\em parallel edges} and this definition extends to paths, which we call {\em parallel paths}. 

\noindent Following this abstract view, a set of triples can be a disconnected graph, with or without isolated nodes, with or without cycles, with or without parallel paths; and it can be a tree, a path, a fork and so on. Moreover a singleton $\{\langle X, \iota_X, X \rangle\}$ is a graph called a {\em trivial graph} and the empty set of triples is also a graph called the {\em empty graph}. Finally, the result of any set theoretic operation on graphs (seen as sets of triples) is also a graph. Motivated by these remarks, given two graphs $\mathcal G$ and $\mathcal G'$, we shall denote by $\mathcal G \sqcup \mathcal G'$ the union of their sets of triples that is: $\mathcal G \sqcup \mathcal G'= Triples(\mathcal G) \cup Triples(\mathcal G')$.




\subsection{Tree restriction}

Before we define the operation of tree restriction we recall that, in our model, the nodes of a database tree represent datasets and the edges represent total functions; and that the edges in a tree are arranged either in a path or in a fork of paths (i.e. a set of paths having the same source). 
Now, the question is: if a node in a path or in a fork is restricted to one of its subsets, how can the restriction be `propagated' to the remaining nodes and edges such as function totality is maintained? To answer this question, let $\mathcal T$ be a tree, $P$ a path in $\mathcal T$ and $F$ a fork of two paths in $\mathcal T$ defined as follows:\smallskip

$P: X \xrightarrow{f} Y \xrightarrow{g} Z$ \smallskip

$F: Z_1\xleftarrow{g_1}Y_1\xleftarrow{f_1}X \xrightarrow{f_2} Y_2 \xrightarrow{g_2} Z_2$ \smallskip

\noindent We shall answer the question separately, first for $P$ and then for $F$. Regarding $P$, we distinguish three cases of restriction propagation as follows: 
\begin{enumerate}
    \item {\em Restriction of the root} 
    
    If $V\subseteq X$ then replace $X$ by $V$, $Y$ by $f(V)$, $f$ by $f/V$ and $g$ by $g/f(V)$
    
    We refer to this action as {\em pushing $V$ to $Z$ along $P$}
    \item {\em Restriction of the target} 
    
    If $V\subseteq Z$ then replace $Z$ by $V$, $Y$ by $g^{-1}(V)$, $X$ by $f^{-1}(g^{-1}(V))$, $f$ by $f/f^{-1}(g^{-1}(V))$ and $g$ by $g/g^{-1}(V)$ 
    
    We refer to this action as {\em pulling $V$ back to $X$ along $P$}
    \item {\em Restriction of an intermediate node} 
    
    If $V\subseteq Y$ then replace $Y$ by $V$, $X$ by $f^{-1}(V)$, $Z$ by $g(V)$, $f$ by $f/f^{-1}(V)$ and $g$ by $g/V$
\end{enumerate} 

Note that, in this case, the result of restriction is independent of the order in 

which pulling back and pushing forward are performed. \smallskip 

\noindent Regarding $F$, let's call $P_1$ the path from $X$ to $Z_1$ and $P_2$ the path from $X$ to $Z_2$. We distinguish four cases as follows: 
\begin{enumerate}
    \item {\em Restriction of the source} 
    
    If $V\subseteq X$ then replace $X$ by $V$, push $V$ to $Z_1$ along $P_1$ and to $Z_2$ along $P_2$
    \item {\em Restriction of a target} (say $Z_1$) 
    
    If $V\subseteq Z_1$ then replace $Z_1$ by $V$, pull $V$ back to $X$ along $P_1$, push $f_1^{-1}(g_1^{-1}(V))$ to $Z_2$ along $P_2$  
    \item {\em Restriction of an intermediate node} (say $Y_1$)  

    If $V_1\subseteq Y_1$ then replace $Y_1$ by $V_1$; replace $g_1$ by $g_1/V_1$; 
    pull $V_1$ to $X$ along $f_1$; push $f_1^{-1}(V_1))$ to $Z_2$ along $P_2$ 
  
    \item {\em Restriction of both targets}  

    If $V_1\subseteq Z_1$ and $V_2\subseteq Y_2$ then 
    
    (a) pull $V_1$ back to $X$ along $P_1$ and let $H_1\subseteq X$ be the result of this pullback 
    
    (b) pull $V_2$ back to $X$ along $P_2$ and let $H_2\subseteq X$ be the result of this pullback 
    
    (c) push $H_1\cap H_2$ to $Z_1$ along $P_1$ and push $H_1\cap H_2$ to $Z_2$ along $P_2$ 

    Note that, in this case, as set intersection is commutative, the result of restriction is independent of the order in which the pullbacks in (a) and (b) above are performed.
\end{enumerate} 

\noindent A few remarks are in order here: 
(a) restriction propagation as described above can be extended in a straightforward manner to a path with more than two edges and to a fork with more than two paths - and this means that it can be extended to a tree; (b) restriction propagation always terminates and the result is independent of the order in which the actions of pulling and pushing are performed; 
(c) the functions represented by the edges remain total in the result of restriction propagation; (d) if the path $P$ consists of only two edges and we restrict the source then restriction propagation reduces to  restriction of function composition; and if a fork $F$ consists of only two edges and we restrict the source then restriction propagation reduces to the restriction of a function pairing. 


\begin{definition}[Tree restriction]
    Let $\mathcal T$ be a database tree with root $\rho$, let $X$ be a node of $\mathcal T$ and let $V$ be a subset of $X$. Then the restriction of $\mathcal T$ to $V$, denoted by $\mathcal T_{V\subseteq X}$ is the tree defined as follows:

\begin{enumerate}
    \item if $X$ is the root of $\mathcal T$ (i.e. if $X= \rho$)
    
    then for every leaf $L$ of $\mathcal T$, push $V$ to $L$ along the unique path from $\rho$ to $L$ 
    \item if $X$ is an intermediate node of $\mathcal T$, $P$ the unique path from $\rho$ to $X$ and $W$ the result of pulling $V$ back to $\rho$ along $P$ 
    
    then push $W$ to every leaf $L$ along the unique path from $\rho$ to $L$
    \item if $X$ is a leaf of $\mathcal T$, say $L$, and $P$ the unique path from $\rho$ to $L$ 
    
    then pull $V$ back to $\rho$ along $P$ and if $W$ is the result of this pulling back then push $W$ to every leaf $L$
\end{enumerate} 
\end{definition} 


\noindent Clearly when applying the above definition there are redundant occurrences of pulling back and pushing forward. For example, suppose $X$ is an intermediate node (item 2 in the above definition) and let $P$ be the unique path from $\rho$ to $X$. Suppose now that $X$ is the common source of two paths $P_1$ and $P_2$ leading from $X$ to leaves $L_1$ and $L_2$, respectively. 
Then pushing $W$ from $\rho$ to leaves $L_1$ and $L_2$ actually pushes $W$ to $X$ along $P$ twice (in general, as many times as there are leaves that are descendants of $X$). Hence the need for efficient algorithms for performing tree restriction. However, designing such algorithms lies outside the scope of the present paper and it is part of our current research.

\subsection{Tree composition}  
Intuitively, given two database trees $\mathcal T$ and $\mathcal T'$, we would like to `combine' them so as to create a `larger' tree thus enabling the extraction of more complex information from the database. In view of our previous discussion, the obvious candidate for doing so is the union $\mathcal T \sqcup \mathcal T'$. Unfortunately this union is not always a tree, as it may create cycles and/or parallel paths, disconnected components and so on. Therefore we must impose conditions to ensure that this union is indeed a tree.  

\begin{definition}[Composability]
 $\mathcal T$ is called {\em composable} with $\mathcal T'$ if $\mathcal T$ and $\mathcal T'$ share only one node, and moreover this node is the root of $\mathcal T'$. We call this unique shared node the {\em link} (of $\mathcal T$ to $\mathcal T'$).   
\end{definition} 


\noindent It is important to note that the link can be {\em any} node of $\mathcal T$ that is, it can be a leaf, the root or an intermediate node of $\mathcal T$. The following lemma states how two trees that are composable can be combined so as to return a single tree.
\begin{lemma} 
Let $\mathcal T$, $\mathcal T'$ be two trees in a database. If $\mathcal T$ is composable with $\mathcal T'$ then $\mathcal T \sqcup \mathcal T'$ is a tree called the {\em composition} of $\mathcal T$ with $\mathcal T'$, denoted by $\mathcal T' \circ \mathcal T$. 
\end{lemma}
\proof Let $\rho$ be the root of $\mathcal T$. As $\mathcal T$ is a tree, there is a unique path from $\rho$ to every node of $\mathcal T$, and in particular to the link node, say $X$. Now, as $X$ is the root of $\mathcal T'$, there is a unique path from $X$ to every node of $\mathcal T'$, and as $X$ is the only node shared by $\mathcal T$ and $\mathcal T'$, it follows that there is a unique path from $\rho$ to every node of $\mathcal T'$. More precisely, let $P$ be the unique path in $\mathcal T$ from $\rho$ to $X$, let $N$ be a node of $\mathcal T'$ and let $P'$ be the unique path in $\mathcal T'$ from $X$ to $N$. Then the `concatenation' of $P$ and $P'$ is the unique path from $\rho$ to $N$. It follows that there is a unique path from $\rho$ to every node of $\mathcal T$ and to every node of $\mathcal T'$, therefore $\mathcal T' \circ \mathcal T$ is a tree whose root is that of $\mathcal T$. \\
\smallskip
\\\noindent Hereafter, motivated by the proof of the above lemma, if $\mathcal T$ is composable with $\mathcal T'$, we shall call {\em composition path} the unique path that exists from the root $\rho$ of $\mathcal T$ to the link $X$ (i.e. to the root $\rho'$ of $\mathcal T'$).  
\noindent Note that the composition of two functions, seen as separate one-edge trees, $\mathcal T: X\xrightarrow{f} Y$ and $\mathcal T': Y\xrightarrow{g} Z$, is a special case of tree composition with the link being the node $Y$ and $f$ being the composition path. Therefore the use of the symbol `$\circ$' to denote the composition of two trees should create no confusion. Similarly, the pairing of two functions with the same source, seen as separate trees, $\mathcal T: Y\xleftarrow{f} X$ and $\mathcal T': X\xrightarrow{g} Z$ is a special case of tree composition with the link being the node $X$ and the identity path $\iota_X$ being the composition path. \smallskip

\noindent Now, if two trees $\mathcal T$ and $\mathcal T'$ are in the same database and $\mathcal T$ is composable with $\mathcal T'$ then the question is: how to define the instance of $\mathcal T \circ \mathcal T'$ in terms of the instances of $\mathcal T$ and $\mathcal T'$ (so as to be able to answer queries against their composition). 
Our approach to answer this question is motivated by the example of composition of two total functions seen as two separate trees, $\mathcal T: X \xrightarrow{f} Y$ and $\mathcal T': Y \xrightarrow{g} Z$, and residing in the same database. Then $\mathcal T$ is composable with $\mathcal T'$ and the link is $Y$. Now, in general, we have $\delta_{\mathcal{T}}(Y) \neq \delta_{\mathcal{T}'}(Y)$, therefore we can distinguish three possible cases for performing the composition $\mathcal{T'}\circ \mathcal{T}$ (i.e. $g \circ f$ in our example), while maintaining the totality of the functions $f$ and $g$: 
\begin{itemize}
    \item The two instances, $\delta_{\mathcal{T}}(Y)$ and $\delta_{\mathcal{T}'}(Y)$ have an empty intersection, in which case the result of the composition $g \circ f$ is the empty tree
    \item The two instances are equal, in which case we have the usual composition of two total functions 
    \item The two instances are neither disjoint nor equal. In this case, objects lying outside the intersection $I= \delta_{\mathcal{T}}(Y) \cap \delta_{\mathcal{T}'}(Y)$ are not `useful' (as they don't participate in the composition) so we can `forget' them. Therefore, we can restrict our attention to $I$, and as $I$ is a restriction of $Y$, we can pull it back to $X$ along $f$ and push it forward to $Z$ along $g$ to obtain the `useful part' of the composition that is $g'\circ f'$, where $f'= f/f^{-1}(I)$ and $g'=g/I$ are total functions. 
\end{itemize} 
\noindent The important thing to note here is that if the database contains only the two trees $\mathcal T$ and $\mathcal T'$ of our previous example then the path expression $g\circ f$ has a meaning {\ only if} expressed against the tree $\mathcal T' \circ \mathcal T$ (and its value will be $g' \circ f'$). In other words, as we mentioned earlier, composition of trees creates paths that traverse two or more trees, thus enriching the query language that we shall see in the following sections. \smallskip

\noindent Let us now add one more edge to the tree $\mathcal T$ of our example that is let us consider the following two trees: $\mathcal T: W \xleftarrow{h} X \xrightarrow{f} Y$ and $\mathcal T': Y \xrightarrow{g} Z$. Clearly the two trees remain composable, and in order to find the instance $\delta_{\mathcal T \circ \mathcal T'}$ in terms of the instances of $\mathcal T$ and $\mathcal T'$, we can follow the same reasoning as above, adding one more step to restriction propagation, namely: replace $h$ by $h'= h/f^{-1}(I)$. \smallskip

\noindent We summarize our discussion so far in the following definition of instance of tree composition.
\begin{definition}[Instance of Tree Composition]
    Let $\mathcal T,\mathcal T'$ be two trees residing in the same database with roots $\rho$ and $\rho'$, and instances $\delta_T$ and $ \delta_{T'}$, respectively. If $\mathcal T$ is composable with $\mathcal T'$ with link $X$
    then the instance $\delta_{\mathcal T \circ \mathcal T'}$ is defined as follows: 
   
\noindent {\em \bf Step 1} If $I=\delta_\mathcal T(X)\cap \delta_{\mathcal T'}(\rho')= \emptyset$ then $\delta_{\mathcal T \circ \mathcal T'}= \emptyset$ else 

\noindent {\em \bf Step 2} Pull $I$ backwards from $X$ to $\rho$ along the composition path and let $V$ be the resulting restriction of $\delta_\mathcal T(\rho)$. Then 

\noindent {\em Step 2.1} Push $\delta_\mathcal T(\rho)$ from $\rho$ to each leaf $L$ of $\mathcal T$ along the unique path from $\rho$ to $L$

\noindent {\em Step 2.2} Push $I$ from $\rho'$ to each leaf $L'$ of $\mathcal T'$ along the unique path from $\rho'$ to $L'$ 
\end{definition}

\noindent What this definition actually says is the following: to find the instance $\delta_{\mathcal T \circ \mathcal T'}$, first restrict $\delta_\mathcal T$ with $I\subseteq \delta_\mathcal T(X)$ (this is step 2.1) then restrict $\delta_{\mathcal T'}$ with $I\subseteq \delta_{\mathcal T'}(\rho')$ (this is step 2.2).
Clearly, if in step 1 $I= \delta_\mathcal T(X)= \delta_{\mathcal T'}(X)$ then no restriction propagation is needed that is the nodes and edges of $\delta_{\mathcal T \circ \mathcal T'}$ keep the values they have in $\delta_\mathcal T$ and $ \delta_{\mathcal T'}$.
Note that pulling back and pushing forward in the above definition maintains the constraint of function totality for each edge of $\mathcal T' \circ \mathcal T$ as required by our definition of tree instance (see Definition \ref{TrIn}).

\noindent As a last remark regarding tree composition, this operation is clearly non-commutative and non-associative, in general. However, in particular cases, composition can be commutative and/or associative. For example, when $\mathcal T, \mathcal T'$ share only their root 
then tree composition is commutative; and if moreover we have three trees, $\mathcal T, \mathcal T', \mathcal T''$ sharing only their root, then composition is also associative. 

\subsection{Cartesian product of trees}
We have seen that tree composition combines two trees that share a unique node (the link) so as to create a larger tree. 
On the other hand, in a tree database, we may be interested in a collection of trees which, although pairwise disjoint (i.e. no node in common) they contain related information. For instance, a tree database may contain a tree $\mathcal T_{Emp}$ describing characteristics of employees and a tree $\mathcal T_{Trans}$ describing characteristics of the company's transportation means. Although these two trees may share no node, we might want to be able to refer to their pair as subtrees of a larger tree, say $\mathcal T_{Admin}$, whose root serves as the entry point for the pair of trees. The Cartesian product serves precisely  this purpose. This operation `puts together' two disjoint trees and returns a tree having the input trees as sub-trees. It is a derived operation defined in terms of tree composition and Cartesian product of sets. 
\begin{definition}
    Let $\mathcal T$ and $\mathcal T'$ be two disjoint trees with roots $X$ and $Y$, respectively. Let $\mathcal T_X: X\times Y \xrightarrow{\pi_X} X$ and $\mathcal T_Y: X\times Y \xrightarrow{\pi_Y} Y$ be the projections of $X\times Y$ over $X$ and $Y$, respectively, seen as separate trees (with one edge each). Then the {\em product} of $\mathcal T$ and $\mathcal T'$ is a tree denoted by $\mathcal T\times \mathcal T'$ and defined by: $\mathcal T\times \mathcal T'= (\mathcal T_X\circ \mathcal T) \circ (\mathcal T_Y\circ \mathcal T')$.
The current instance of $\mathcal T\times \mathcal T'$ is computed from the current instances of $X\times Y$, $\pi_X$, $\pi_Y$, $\mathcal T_X$ and $\mathcal T_Y$, by applying the operations used in the definition of $\mathcal T\times \mathcal T'$.
\end{definition}

\noindent To see that $\mathcal T\times \mathcal T'$ is well defined, observe that $\mathcal T_X$ is composable with $\mathcal T$ (with link $X$), $\mathcal T_Y$ is composable with $\mathcal T'$ (with link $Y$) and $(\mathcal T_X\circ \mathcal T)$ is composable with $(\mathcal T_Y\circ \mathcal T')$ (with link $X\times Y$). The definition of $\mathcal T\times \mathcal T'$ is depicted in Figure \ref{Fig-5}.

\noindent Note that the Cartesian product of trees is a commutative and transitive operation over mutually disjoint trees. 
Also note that the Cartesian product of trees can be extended to non-disjoint trees if we prefix each node shared by the two trees by the name of the tree to which it belongs. For example, the trees $\mathcal T: X \xrightarrow{f} Y$ and $\mathcal T': Y \xrightarrow{g} Z$ are not disjoint as they share the node $Y$. However, if we replace $Y$ by two nodes, $\mathcal T.Y$ and $\mathcal T'.Y$ such that $dom(\mathcal T.Y)= dom(\mathcal T'.Y)$ then the two trees become disjoint and we can take their Cartesian product whose root is $(X\times \mathcal T'.Y)$ with $X\to \mathcal T.Y$ and $\mathcal T'.Y\to Z$ as its sub-trees. This `renaming' technique is well-known in relational databases, where it is used for renaming attributes which are shared by two relation schemas (in particular, when defining `joins' \cite{Ullman83}). 

\noindent As a last remark, the product $f\times g$ of two functions $f:X\to Y$ and $g:X'\to Y'$ that we have seen earlier (see Section \ref{subsec:FA}, Lemma \ref{Prod}) is a special case of the Cartesian product of trees, as stated in the following lemma. 
\begin{lemma}
Let $f: X\to Y$ and $g: X'\to Y'$ be two functions and let $\mathcal T: X\xrightarrow{f} Y$ and $\mathcal T': X'\xrightarrow{g} Y'$ be the functions $f$ and $g$ seen as trees. Let $\mathcal T_X: X\times X'\xrightarrow{\pi_X}\ X$, $\mathcal T_{X'}  : X\times X'\xrightarrow{\pi_{X'}}\ X'$ be the projections of $X\times X'$ also seen as trees. Then we have: 
$(\mathcal T\circ \mathcal T_X)\circ (\mathcal T'\circ \mathcal T_{X'})$ is the product function $f\times g$ seen as a tree. 
\end{lemma}
\proof First observe that $\mathcal T_X$ is composable with $\mathcal T$ (with link $X$) and $\mathcal T_{X'}$ is composable with $\mathcal T'$ (with link $X'$). Moreover, $\mathcal T\circ \mathcal T_X$ is composable with $\mathcal T'\circ \mathcal T_{X'}$ (with link $X\times X'$). As $X\times X'$ is the common root of $\mathcal T\circ \mathcal T_X$ and $\mathcal T'\circ \mathcal T_{X'}$, it follows that $(\mathcal T\circ \mathcal T_X)\circ (\mathcal T'\circ \mathcal T_{X'})= (\mathcal T\circ \mathcal T_X)\wedge (\mathcal T'\circ \mathcal T_{X'})$ which is the product function $f\times g$ seen as a tree. \smallskip

\subsection{Tree projection}\label{subsec:TP}
Restriction, composition and Cartesian product are tree operations each taking as input one or more trees and creating a new tree. In the opposite direction, it is desirable to be able to `decompose' a given tree into two composable subtrees such that their composition reconstructs the original tree.
For example when designing a tree database one may want to `experiment' different ways of configuring the information by either composing two or more trees in order to create a larger tree or by decomposing a tree into smaller subtrees so that the information in the database is more focused and easier to grasp. To this end we need operations that take as input a tree $\mathcal T$ and a node $X$ of $\mathcal T$ and return a subtree of $\mathcal T$ rooted at $X$. Therefore we introduce two such operations, `tree projection' and `star tree partition'.

\begin{definition}[Tree projection]
 Let $\mathcal T$ be a tree and $X$ a node of $\mathcal T$. Then the {\em projection} of $\mathcal T$ on $X$, denoted by $\pi_X(\mathcal T)$, is the subtree of $\mathcal T$ defined as follows: 
\begin{itemize}
    \item the set of nodes consists of $X$ and all descendants of $X$ in $\mathcal T$ 
    \item the set of edges consists of all edges of $\mathcal T$ connecting nodes in the following set: 
    
    $\{X\}\cup \{Y/ Y \mbox{ is a descendant of $X$ in } \mathcal T\}$ 
\end{itemize}
\end{definition} 

\noindent To see that $\pi_X(\mathcal T)$ is indeed a subtree of $\mathcal T$ it is sufficient to note that (a) there is a unique path from the root of $\mathcal T$ to every node of $\pi_X(\mathcal T)$ (because of item 1 above), (b) there is a unique path from the root of $\mathcal T$ to $X$ and (c) every node of $X$ is a descendant of $X$ in $\mathcal T$. Therefore there is a unique path from $X$ to every node of $\pi_X(\mathcal T)$. Note that if $X$ is the root of $\mathcal T$ then $\pi_X(\mathcal T)= \mathcal T$ and if $X$ is a leaf of $\mathcal T$ then $\pi_X(\mathcal T)$ is a trivial subtree (i.e. it consists of the node $X$ and its identity edge $\iota_X$). Also note that the use of $\pi_X$ to denote both, projection of the Cartesian product of sets and tree projection should create no confusion as the function $\pi_X$ applies on arguments of different type in each case (Cartesian products of sets in the first case versus trees in the second case). \smallskip

\noindent Now, it is not difficult to see that what remains of $\mathcal T$ if we remove from it all edges of $\pi_X(\mathcal T)$ is also a subtree of $\mathcal T$. To see this, note that there is a unique path from the root of $\mathcal T$ to $X$ and also from the root of $\mathcal T$ to every node not in $\pi_X(\mathcal T)$. We shall refer to this subtree as the {\em complement} of $\pi_X(\mathcal T)$ and if we set $T_X= \pi_X(\mathcal T)$ then we shall denote its complement by  $\overline{\mathcal T}_X$. Therefore we have: $\overline{\mathcal T}_X= triples(\mathcal T)\setminus triples(\mathcal T_X)$. Figure \ref{Fig-5}(b) shows the projection of a tree and its complement. 

\noindent It is important to note that (a) $T_X$ and $\overline{\mathcal T}_X$ are composable (with the node $X$ as the link) and (b) they `make up' the whole tree $\mathcal T$ in the sense that $\overline{\mathcal T}_X\circ T_X= \mathcal T$. Therefore, using tree projection, we can decompose a tree $\mathcal T$ into two composable subtrees as stated formally in the following definition.

 \begin{figure}
{
\begin{center}
\includegraphics[width=350px,keepaspectratio]{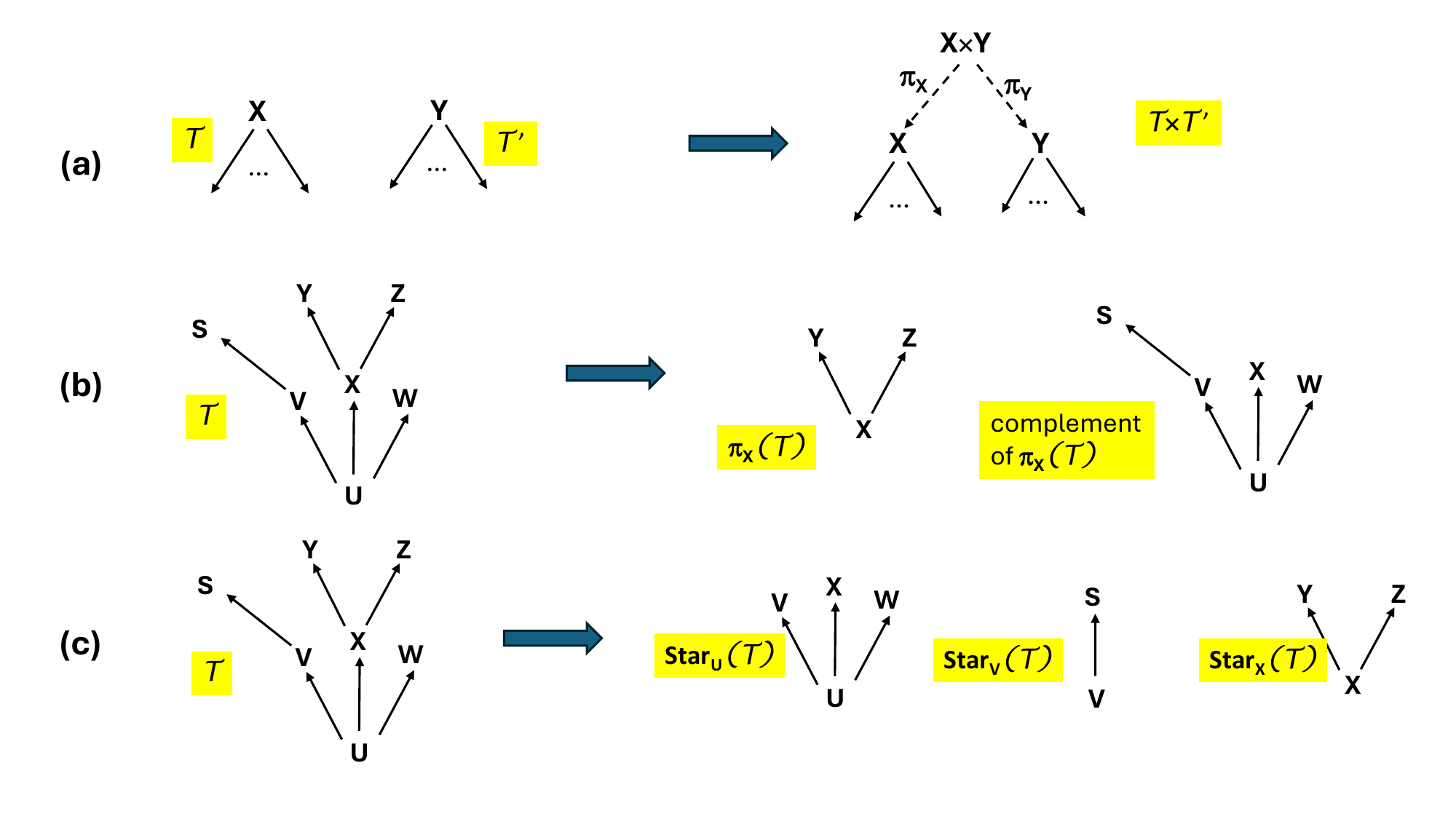}
\caption{(a) Product of trees  (b) Projection of a tree on a node $X$ (c) The star subtrees of of $\mathcal T$ \label{Fig-5}}
\end{center}
}
\end{figure}

\begin{definition}[Tree projection]
Let $\mathcal T$ be a tree and $X$ a node of $\mathcal T$. 
Then the pair $(\mathcal T_X, \overline{\mathcal T}_X)$ is called the {\em decomposition} of $\mathcal T$ by projection on $X$.
\end{definition} 
 

\noindent Now, given an instance $\delta_\mathcal T$ of $\mathcal T$, the question is what are the instances of the components in a decomposition of $\mathcal T$ in terms of $\delta_\mathcal T$. Roughly speaking, the instances of the components are the components of $\delta_\mathcal T$. More formally we have:
\begin{definition}
    Let  $(\mathcal T_X, \overline{\mathcal T}_X)$ be a decomposition of a tree $\mathcal T$ and let $\delta_\mathcal T$ be an instance of $\mathcal T$. Then we have:

    \noindent - if $N$ of $\mathcal T_X$ then $\delta_{\mathcal T_X}(N)= \delta_\mathcal T(N)$ and if $e$ is an edge of $\mathcal T_X$ then $\delta_{\mathcal T_X}(e)= \delta_\mathcal T(e)$ 
    
    \noindent - if $N$ is a node of $\overline{\mathcal T}_X$ then $\delta_{\overline{\mathcal T}_X}(N)= \delta_{\mathcal T}(N)$ and if $e$ is an edge 
    of $\overline{\mathcal T}_X$ then $\delta_{\overline{\mathcal T}_X}(e)= \delta_{\mathcal T}(e)$
\end{definition} 

\noindent Suppose now that, in a database, we have a tree $\mathcal T$ with instance $\delta_\mathcal T$ and we want to replace it by the two trees of the decomposition $(\mathcal T_X, \overline{\mathcal T}_X)$, where $X$ is a node of $\mathcal T$. Then we have two cases: 

\noindent - if we want to let the two trees evolve independently during database updates then we may have $\delta_{\mathcal T_X}(X)\ne \delta_{\overline{\mathcal T}_X}(X)$

\noindent - if we want to be able to `recover' $\delta_\mathcal T$ (by composition) in every database instance then we must maintain the set-theoretic constraint: $\delta_{\mathcal T_X}(X)= \delta_{\overline{\mathcal T}_X}(X)$ \\

\noindent To illustrate the above statement consider the database tree: $\mathcal T: X\xrightarrow{f} Y\xrightarrow{g} Z$

\noindent If we decompose $\mathcal T$ by projection on $Y$ then we obtain the decomposition $(\mathcal T_Y, \overline{\mathcal T}_Y)$, where $\mathcal T_Y= g$ and $\overline{\mathcal T}_Y= f$. Now, if we replace $\mathcal T$ by this decomposition then: 

\begin{itemize}
    \item if we do not impose any constraint then at some point in time we may have: 

$\delta_{\mathcal T_Y}(Y)\ne \delta_{\overline{\mathcal T}_Y}(Y)$ hence $\delta_\mathcal T\ne \delta_{\mathcal T_Y}(g)\circ \delta_{\overline{\mathcal T}_Y}(f)$

\item else if we impose the set theoretic constraint $\delta_{\mathcal T_Y}(Y)= \delta_{\overline{\mathcal T}_Y}(Y)$ 

then we will always have $\delta_\mathcal T=\delta_{\mathcal T_Y}(g)\circ \delta_{\overline{\mathcal T}_Y}(f)$
\end{itemize}







    


\subsection{Star tree partition}\label{subsec:ST}
We recall first that a star tree is a tree in which all leaves are (immediate) successors of the root. Now, given any tree $\mathcal T$, the root of $\mathcal T$ together with the set of edges from the root to its successors constitutes a star tree; and this holds for each successor of the root and, recursively, for all nodes of $\mathcal T$ (with the leaves of $\mathcal T$ being the roots of trivial star trees). Given a tree $\mathcal T$ and a node $X$ of $\mathcal T$, we shall denote by $Star_X(\mathcal T)$ the star subtree of $\mathcal T$ rooted at node $X$. Figure \ref{Fig-5}(c) shows a tree $\mathcal T$ and all its star subtrees. It is not difficult to see that we can reconstruct $\mathcal T$ from its star subtrees. For example, referring to Figure \ref{Fig-5}(c), we have: $T= Star_X(T)\circ (Star_V(T)\circ Star_U(T))$. Actually, any tree $\mathcal T$ can be seen as the composition of the star subtrees defined by its (non-leaf) nodes.
\smallskip 

\noindent Now, a star tree can be partitioned into two star trees by partitioning its set of edges as expressed in the following definition.
\begin{definition}[Star tree partition]
 Let $\mathcal T$ be a star tree with root $\rho$, let $L$ the set of leaves of $\mathcal T$ and let $S$ be a set of leaves of $\mathcal T$ (i.e. $S\subseteq L$). Then the {\em partition} of $\mathcal T$ on $S$, denoted by $Part_S(\mathcal T)$ is the star tree with root $\rho$ and with the following set of edges:
 $\{\rho\to \lambda/ \lambda\in S\}$.  
\end{definition}

\noindent As in the case of projection, partitioning applied to a star tree produces a decomposition of the star tree into two star trees: $Part_S(\mathcal T)$ and $Part_{L\setminus S}(\mathcal T)$, the latter called the {\em complement} of $Part_S(\mathcal T)$. Moreover, given an instance $\delta_\mathcal T$ of $\mathcal T$ the instances of $Part_S(\mathcal T)$ and its complement are computed in a similar way as in the case of tree projection. 

\subsection{Tree algebra}\label{subsec:TA}
We have seen so far several operations on trees. In the rest of the paper, we shall refer to the set of all these operations as the {\em tree algebra} and to a well formed expression using these operations as a {\em tree expression}. 

\begin{definition}[Tree algebra and Tree expression]
Let $\mathcal D$ be a tree database. 
\begin{itemize}
    \item 
The {\em tree algebra} of $\mathcal D$ is defined to be the set of the following operations on trees: restriction, composition, projection and star tree partition (the Cartesian product of trees being a derived operation). 
\item A {\em tree expression} over $\mathcal D$ is either a tree of $\mathcal D$ or a well formed expression $TE$ whose operands are trees of $\mathcal D$ and whose operations are among those of the tree algebra. The value of a tree expression $TE$, denoted by $Val_{TE}$, is the tree obtained by performing the operations in $TE$; and given the current instance of $\mathcal D$, the current instance of $Val_{TE}$ is obtained in the way that we have seen earlier in this section. 
\end{itemize} 
\end{definition}
\noindent To illustrate the concept of tree expression refer to Figure \ref{Fig-1} and consider the tree expression $T_2\circ T_1$. Its value is the tree $T_5$ and the (current) instance of $T_5$ is computed from the instances of $\mathcal T_1$ and $\mathcal T_2$ as explained for tree composition earlier in this section. As another example, referring to Figure \ref{Fig-5}(c), the following tree expressions have both the same value, namely the tree $\mathcal T$: \smallskip 

\noindent $Star_X(T)\circ (Star_V(T)\circ Star_U(T))$ ~~ and ~~ $Star_X(T)\circ (Star_U(T)\circ Star_V(T))$ \smallskip

\noindent By the way, two expressions (such as the above) having the same value are called {\em equivalent}. 

\noindent Actually, a tree expression $TE$ over a tree database $\mathcal D$ can be seen as a `tree query' over $\mathcal D$, whose answer is the value $Val_{TE}$ of $TE$.

\noindent Note that the operations of the tree algebra allow for what we could call `tree games': given two databases, $\mathcal D$ and $\mathcal D'$, use the operations of the tree algebra to produce $\mathcal D'$. For example, referring to Figure \ref{Fig-1} and starting with the database $\mathcal D= \{T_1, T_2, T_3\}$ we can produce the database $\mathcal D'= \{T_5, T_6\}$ using the following tree expressions: $\mathcal T_5= \mathcal T_2\circ \mathcal T_1$ ~and~ $\mathcal T_6= \mathcal T_3\circ \mathcal T_1$ 

\noindent Conversely, starting with $\mathcal D'$ we can produce $\mathcal D$ using the following tree expressions:

$\mathcal T_1= Star_{Inv}(T_6)$ ~~and ~ $\mathcal T_2= \pi_{Branch}(\mathcal T_5)$

\noindent By the way, two databases, $\mathcal D$ and $\mathcal D'$, such that starting with $\mathcal D$ we can produce $\mathcal D'$, and conversely, starting with $\mathcal D'$ we can produce $\mathcal D$, are called {\em structurally equivalent}. \smallskip

\noindent We note in passing that two extreme cases of such `tree games' are (a) when $\mathcal D$ is any database and $\mathcal D'$ contains a single tree and (b) when $\mathcal D$ contains a set of one-edge trees (i.e. a set of edges called {\em tokens}) and $\mathcal D'$ contains a single tree.  The first case is related to the `universal relation' concept in relational databases \cite{Ullman83}), a concept that motivated a considerable amount of research work as well as much controversy during the late 1970s and 1980s over whether database queries and design should be based on a single, hypothetical table containing every attribute (the `Universal Relation') or traditional normalized tables \cite{Mendelzon}.  

\noindent The second case is related to puzzles, where one tries to compose tokens in an effort to produce a single tree - which is usually an image. 
We shall come back to these remarks when discussing tree database design in Section \ref{sec: Conclusions}.


\smallskip


\noindent As a last remark, the operations of the tree algebra contribute to increase the expressive power of our model in at least three significant ways: 

\noindent First, as we shall see in the following sections, by creating larger trees, the operations of composition and Cartesian product enrich the query language by offering the possibility of expressing new queries traversing two or more trees. 

\noindent Second, by decomposing a large tree into smaller trees, tree decomposition allows to focus on smaller parts of the database. 

\noindent Third, as we shall see in the following subsection, using these operations we can define quite sophisticated views of a tree database tailored to the needs of individual users or user groups.

\subsection{Views}\label{subsec:Views}

A tree expression $TE$ over a tree database $\mathcal D$ can be seen as a `tree query' over $\mathcal D$, whose answer is the value $Val_{TE}$ of $TE$. The concept of view was introduced in the early days of relational databases \cite{Ullman83} and originated from the need of individual users or groups of users wanting to use only part of the information contained in the database, and moreover tailored to their specific needs. Formally, a view of a relational database is a pair $(\mathcal V, RE)$, where $\mathcal V$ is the view name and $RE$ is a relational algebra expression called the `view definition'. Such a view actually defines a relational table whose name is $\mathcal V$ and whose content is the value of $RE$ \cite{Ullman83}. Motivated by the concept of relational view, we introduce the following definition of view of a tree database.


\begin{definition}[View of a Tree Database]\label{subsec: Views}
 Let $\mathcal D$ be a tree database. A {\em view} of $\mathcal D$ is defined to be a pair $(\mathcal V, TE)$, where $\mathcal V$ is the view name and $TE$ is a tree expression over $\mathcal D$, called the {\em view definition}.
 \end{definition} 


 \noindent To illustrate this definition, consider the tree database $\mathcal D$ of Figure \ref{Fig-2}. The pair $(\mathcal V, \pi_{Prod}(\mathcal T_3)\circ(T_2\circ T_1))$ is a view of $\mathcal D$ whose name is  $\mathcal V$ and whose definition is the expression $T_2\circ T_1$. Incidentally, as the trees $\mathcal T_2$ and $\mathcal T_3$ are disjoint, we have: $\pi_{Prod}(\mathcal T_3)\circ (\mathcal T_2\circ \mathcal T_1)= \mathcal T_2 \circ (\pi_{Prod}(\mathcal T_3)\circ \mathcal T_1)$; and in fact this property holds in general, as stated formally in the following lemma, whose proof follows immediately from the definition of composition.
\begin{lemma}
    Let $\mathcal D$ be a tree database and $\mathcal T_1, \mathcal T_2, \mathcal T_3$ three trees of $\mathcal D$ such that (a) $\mathcal T_1$ is composable with $\mathcal T_2$ and with $\mathcal T_3$ and (b) $\mathcal T_2$ and $\mathcal T_3$ are disjoint. Then we have: $\mathcal T_3\circ (\mathcal T_2\circ \mathcal T_1)= \mathcal T_2 \circ (\mathcal T_3\circ \mathcal T_1)$.
\end{lemma}

As a last remark on views, we have seen earlier that, using tree projection or star tree partition it is possible to decompose a tree into two or more subtrees and this possibility is quite welcome. Indeed, using decompositions when defining a view, a view user may wish to focus on smaller trees serving specific information needs. \smallskip

 \begin{figure}
{
\begin{center}
\includegraphics[width=350px,keepaspectratio]{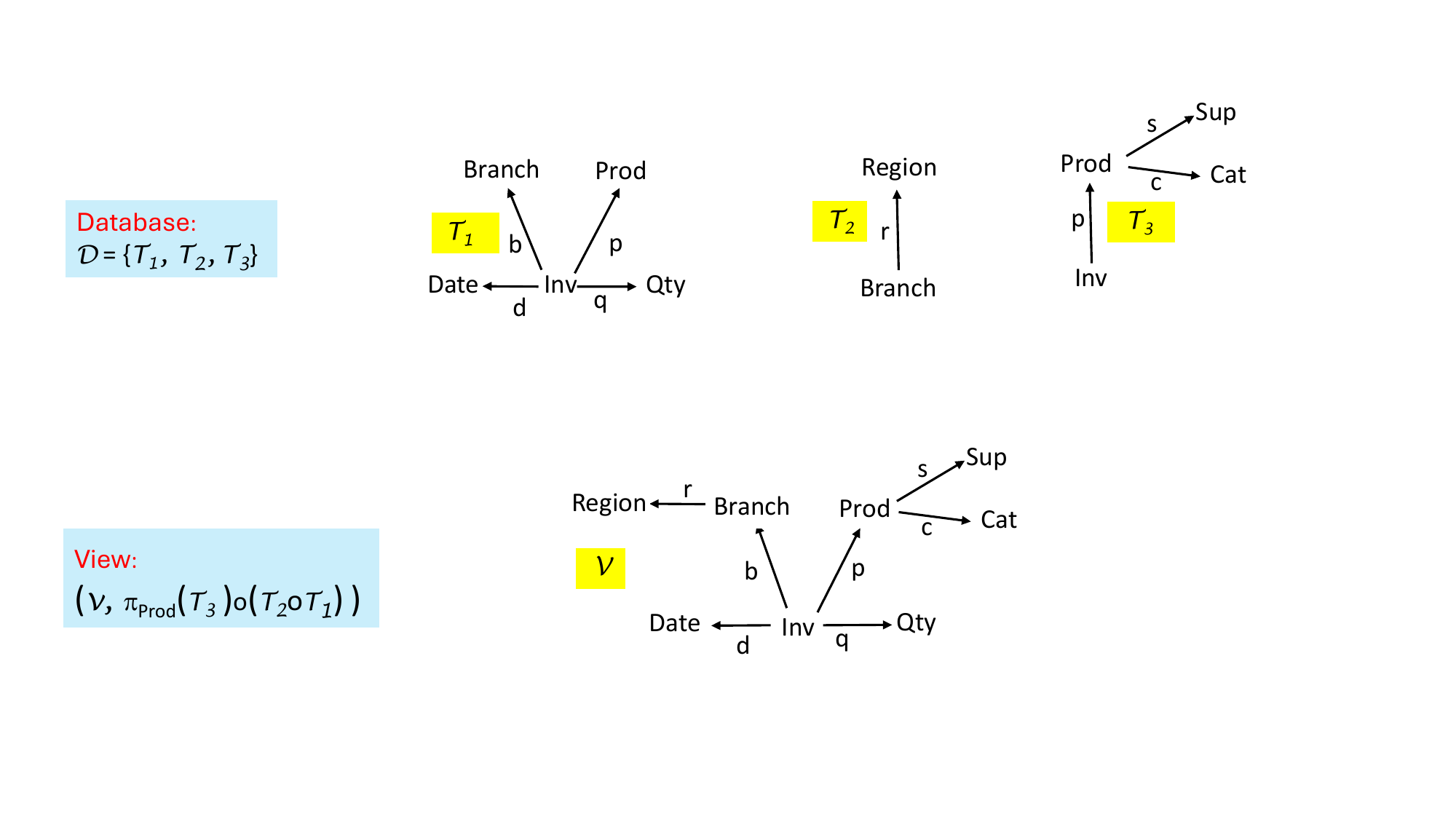}
\caption{Example of view \label{Fig-2}}
\end{center}
}
\end{figure}


\section{Traversal Queries}\label{sec:TQ}

As explained informally in the introduction, our model offers two kinds of queries, traversal queries and analytic queries; and these queries are always expressed over a tree, which is either a tree present in the database or a tree derived from other trees using a tree expression.  Therefore, in our approach, the user interacts with the tree database at two levels: (a) querying the database to extract a desired tree and (b) querying the extracted tree using one of two kinds of queries, traversal queries or analytic queries. 
In this section, given a tree $\mathcal T$, we present the definition of a traversal query over $\mathcal T$, and in the following section we present the definition of an analytic query over $\mathcal T$. \smallskip

\noindent A traversal query is a type of query used to explore and navigate a tree by traversing its nodes and edges, starting from a specific node or a set of nodes and following the connections between nodes based on defined criteria. 
The definition of a traversal query is based on that of path expression over a tree, a concept that we have seen in section \ref{subsec:FA}. We recall that a path expression is the composition of all edges along a path; and that there is a one-one correspondence between paths in a tree and path expressions: given a path, the corresponding path expression is the composition of its edges, and conversely, given a path expression the corresponding path is the sequence of edges in the expression in reverse order. Therefore we can use a path to specify a path expression, unambiguously; and we can use a path expression to specify a path, unambiguously. Clearly, the source and target of a path expression are the source and target of the corresponding path.

\begin{definition}[Traversal query over a tree]\label{section:TQ} 
\noindent Let $\mathcal D$ be a tree database and $\mathcal T$ a tree, which is either in $\mathcal D$ or derived from $\mathcal D$ using a tree expression. A {\em traversal query} $Q$ over $\mathcal T$ is defined to be either a single path expression over $\mathcal T$ or the pairing of two or more path expressions having the same source $S$ and distinct targets $A_1, \ldots , A_n$, that is: 

\noindent $Q= PE_1\wedge \ldots \wedge PE_n$, where $target(PE_1)= A_1, \ldots, target(PE_n)= A_n$, $n\ge 1$ 

\noindent The {\em answer} of $Q$ in an instance $\delta_\mathcal T$ of $\mathcal T$ is denoted by $Ans_Q(\delta_\mathcal T)$ and it is defined to be the pairing of the values of $PE_1, \ldots, PE_n$ in $\delta_\mathcal T$ that is:

$Ans_Q(\delta_\mathcal T)= Val_{PE_1}(\delta_\mathcal T)\wedge \ldots \wedge Val_{PE_n}(\delta_\mathcal T)$

\noindent We shall denote a traversal query also by $Q(S, PE_1, \ldots , PE_n)$ when we want to state explicitly that the common source of its expressions is $S$; and we shall refer to $S$ as the {\em source} of the query. Moreover, we shall denote the answer simply
by $Ans_Q$ when $\delta_\mathcal T$ is understood. 
\end{definition}
Clearly, as the value of each $PE_i$ is a function from $S$ to $A_i$, say $f_i: S\to A_i$, the answer of $Q$ is the pairing $f_1\wedge \ldots\wedge f_n$ (i.e. a function from $S$ to $A_1\times \ldots \times A_n$). For example, referring to tree $\mathcal T_7$ of Figure~\ref{Fig-1}, if $Q= (s\circ p)\wedge (c\circ p)$ then $Ans_Q$ is a function from $Inv$ to $Sup\times Cat$. We shall refer to $A_1\times\ldots \times A_n$ as `the target' of the query and to each $A_i$ as `a target' of the query. \smallskip

\noindent Note that the identity edge of each node is a traversal query (also called the {\em identity query} of the node). Typically, identity queries are used when we want to `read' the content of a node $X$. For example, the identity query $Q=\iota_{Sup}$ will return the set $\{(x, x) / x \in \delta_{\mathcal T}(Sup)\}$ from which we can `extract' the (current) set of suppliers.\smallskip

\smallskip
\noindent Traversal queries have a number of important features that we summarize below: 
 \begin{enumerate}
     \item The answer to a traversal query can be represented as an easy-to-read table. Indeed, if $A_1, \ldots, A_n$ are the targets of the query, then the answer can be represented in a table such that: the rows are indexed by the values of $S$, the columns are indexed by $A_1, \ldots, A_n$, and for each $s$ in $S$, the value $f_i(s)$ is entered in the cell $(s, A_i)$. 
     
     We recall that, in the relational model \cite{Ullman}, a table where the entries are atomic values and in which the values in each column are functionally dependent on the key-values is said to be in `First Normal Form'. Moreover, if we view each expression $PE_i$ of a traversal query as a functional dependency (from the source $S$ of the query to the target $A_i$ of $PE_i$) then the table just defined satisfies all dependencies $f_i: S\to A_i$. Therefore, the answer of a traversal query can be seen as a relation whose key is $S$ (the source of the query) and whose attributes are the targets  $A_1, \ldots, A_n$ of the path expressions defining the query. 

\item Each path expression $PE_i$ of a traversal query over a tree $\mathcal T$ corresponds to the unique path from the source $S$ to the target $A_i$ of the query. Therefore a traversal query $Q$ over $\mathcal T$ can be specified as follows: (a) choose a node $S$ of $\mathcal T$ as the source of $Q$, (b) choose $n$ distinct descendants $A_1, \ldots, A_n$ of $S$ in $\mathcal T$ as the targets of $Q$, and (c) give restrictions (optionally) for $S$ and for each $A_i$. Then the query $Q$ will be the pairing of the (possibly restricted)
path expressions thus defined. This implies that one can use the select-from-where pattern of  SQL in order to define a traversal query $Q$: \smallskip

{\bf select} $\langle A_1, \ldots, A_n\rangle$ {\bf from} $\langle tree-expression \rangle$ {\bf where} $\langle conditions \rangle$ \smallskip 

Here, the root of the tree defined by the tree expression in the {\em from} clause (call it $\mathcal T$) is the source of the query; the nodes $A_1, \dots, A_n$ are the targets of the query; and the conditions in the {\em where} clause define (optionally) restrictions of the nodes $S, A_1, \dots, A_n$.   
Clearly, before query evaluation, the system will have to verify whether the nodes $A_1, \dots, A_n$ given in the {\em select} clause are distinct descendants of node $S$ in $\mathcal T$. For example, referring to Figure \ref{Fig-1} the following statement defines a traversal query over the tree $\mathcal T_7$:  \smallskip

{\bf select} Sup {\bf from} $\pi_{Prod}(\mathcal T_7)$ {\bf where} $(Prod= X)$ and $(Cat= Y)$   \smallskip

Here, the tree expression is the the projection of $\mathcal T_7$ on $Prod$ and the statement defines a traversal query over the tree $\pi_{Prod}(\mathcal T_7)$, returning the set of all suppliers for product $X$ and category $Y$.

\item As a traversal query $Q$ over a tree $\mathcal T$ is the pairing of $n$ path expressions $PE_1, \dots, PE_n$  with a common source $S$ and targets $A_1, \dots, A_n$,  its answer, $Ans_Q$, can be seen as a set of tuples over `attributes'  $S, A_1, \ldots, A_n$, where $A_1= target(PE_1), \ldots, A_n= target(PE_n)$. For example, consider  the query $Q= (r \circ b)\wedge (s \circ p)$ over the tree ${\mathcal T_7}$ of Figure \ref{Fig-1}. The source of $Q$ is $Inv$, the expression $r \circ b$ has $Region$ as target, and  the expression $s \circ p$ has $Sup$ as target. The answer to this query can be seen as a set of tuples over the `attributes' $Inv, Region, Sup$. In other words, each traversal query $Q(S, PE_1, \ldots, PE_n)$ {\em induces} a relation schema $R_Q(S, A_1, \ldots, A_n)$, and for every instance $\delta_{\mathcal T}$ of $\mathcal T$, the answer $Ans_Q(\delta_{\mathcal T})$ {\em induces} a relation $r_Q$ over $R_Q$. 

\noindent Moreover, as $\pi_{A_i} \circ Ans_Q(\delta_{\mathcal T})= f_i$ is a function from $S$ to $A_i$, $i= 1, \ldots, n$, the relation $r_Q$ {\em satisfies} the functional dependencies $f_1, \ldots, f_n$ (thus making the source $S$ of $Q$ the key of $R_Q$). It follows that there is a mapping from traversal queries $Q$ over $\mathcal T$ to relation schemas $R_Q$ such that: (a) for each instance $\delta_{\mathcal T}$, the answer $Ans_Q(\delta_{\mathcal T})$ induces a relation $r_Q$ over $R_Q$ satisfying the functional dependencies $f_1, \ldots, f_n$, and (b) the values of each non-key attribute $A_i$ of $R_Q$ are computed by the expression $PE_i$ having $A_i$ as its target. 
In other words, each expression $PE_i$ of $Q$ provides the {\em semantics} of attribute $A_i$ of $R_Q$, $i= 1, \ldots, n$; and the pairing of the path expressions in the query provides the semantics of the induced relation schema $R_Q$. Moreover, this semantics is unique as there is one and only one path from the source of the query to each of its targets (put differently, there can't be two parallel paths between the source and a target of the query).  
We shall come back to this point in section \ref{subsec:Interface}, where we discuss the definition of a relational database as a set of traversal queries over a tree database. 
\item In view of our discussion in the previous item, in our model, we consider functional dependencies as `first class citizens' for building queries, while relations are derived concepts. This is in sharp contrast with the relational model, in which relations are considered as `first class citizens' and functional dependencies are considered as `constraints' that the relations must satisfy. Put differently, in the relational model, functional dependencies play a static role, whereas in our model they play a dynamic role in the sense that it is the functional dependencies that actually {\em build-up} consistent relations.
\end{enumerate}
Before ending this section, a word on query evaluation. The evaluation of a traversal query can be done directly, based on the definitions of the operations of the functional algebra that appear in the query; or indirectly, either by translating the query as an equivalent query in some query engine (e.g. by translating it as an SQL query) or by rewriting the query.

\noindent The kind of rewriting that we envisage in our model is the process of applying a number of transformations to the original query in order to produce an equivalent optimized one (`optimized' in the sense that the equivalent query allows to reuse previously obtained query results). Such transformations do not depend on the physical state of the system (such as the size of physical structures, the system workload, etc). They are just well-defined rules that specify how to transform a query expression into a logically equivalent one \cite{DBLP:reference/db/Pitoura18d}. 
Here, by `equivalence' between two queries we mean the following: two traversal queries $Q$ and $Q'$ over a tree $\mathcal T$ are equivalent, denoted as $Q \equiv Q'$, if they return the same answer in every database instance. For example, referring to the tree $\mathcal T_7$ of Figure \ref{Fig-1}, the queries $Q= (s\circ p)\wedge (c\circ p)$ and $Q'= (s\wedge c)\circ p$ are equivalent (as composition distributes over pairing). Yet, the first needs two compositions and one pairing for its evaluation whereas the second needs one composition and one pairing. Such rewriting is obviously based on properties of the functional algebra. However,
a complete account of traversal query rewriting lies outside the scope of the present paper and it is part of our current research. 

\section{Analytic queries}\label{sec:AQ}




In order to give a general definition of analytic query over a tree we proceed as for traversal queries that is: we first define what an analytic expression is and then we define an analytic query to be either an analytic expression or the pairing of two or more analytic expressions with the same source. In doing so we build upon our previous work on data analytics \cite{SpyratosS18}\cite{spyratos2023context}.
First, let us see how an analytic expression is defined and evaluated using the tree $\mathcal T_1$ of Figure \ref{Fig-1} as an example. \smallskip

\noindent Suppose that we want to know the total quantity delivered to each branch. This computation needs the following items: 

\noindent - the edges (functions) $b$ and $q$ of $\mathcal T_1$ that return the branch and the quantity appearing on each invoice, respectively; 

\noindent - and the aggregate operation $sum$ (to sum-up the quantities delivered in each branch). 

\noindent The triple $AE=(b, q, sum)$ is an example of what we call an `analytic expression' and `the total quantity delivered by branch' is what we call the `value' of $AE$ in $\mathcal T_1$, denoted by $Val_{AE}(\mathcal T_1)$. \smallskip

\noindent Figure \ref{AE}(a) shows a toy example, where the data set $Inv$ consists of seven invoices, numbered 1 to 7, with each invoice associated with a branch and a quantity by the functions $b$ and $q$, respectively. In order to find the value of $AE$ we proceed in three steps as follows: 

\noindent{\em Grouping}: During this step we group together all invoices referring to the
same branch (i.e. we invert the function $b$). The results of grouping are shown in Figure~\ref{AE}(a) (e.g. invoices 1, 2 refer to branch $Br_1$).




\noindent {\em Measuring}: In each group of the previous step, we find the quantity corresponding to each invoice in the group, using the function $q$ (e.g. the quantities corresponding to $Br_1$ are 200 and 100).




\noindent {\em Aggregation}: In each group of the previous step, we sum up the quantities found (e.g. for branch $Br_1$ we find 200 + 100= 300).




\noindent Then the association of each branch to the corresponding total quantity is the desired result (and this result is the function $Val_{AE}(\mathcal T_1)$ shown in Figure \ref{AE}(a)):

$Br_1 \to 300$, $Br_2 \to 600$, $Br_3 \to 600$ \smallskip



\noindent We view the ordered triple $AE= (b, q, sum)$ shown in Figure \ref{AE}(a) as an analytic expression over $\mathcal T_1$, the function $Val_{AE}: Branch \rightarrow TotQty$ in that same figure as the value of $AE$, and the computations shown in the figure as the expression evaluation process.
Note though that what makes the association of branches to total quantities possible is $(1)$ the fact that the functions $b$ and $q$ have the same source (namely, $Inv$ in this example) and $(2)$ that `sum' is an operation applicable on $q$-values. Also note that {\em TotQty} is a new symbol (i.e. not appearing in the analytic expression) which is necessary for naming the target of the computed function $Val_{AE}(\mathcal T_1)$. 

\noindent The function $b$ that appears first in the triple $(b, q, sum)$ and used in the grouping step is called the {\it grouping function}; the function $q$ that appears second in the triple is called the {\it measuring function}, or the measure; and the function $sum$ that appears third in the triple is called the {\it aggregate operation}. Actually, the triple $(b, q, sum)$ should be regarded as the specification of a data analysis task to be carried out over the data set $Inv$. 

\noindent Note that by exchanging the first two components of this triple we obtain the triple $(q, b, sum)$ which is {\em not} a well formed analytic expression as, although $b$ and $q$ have the same source, the aggregate operation is not applicable on b-values (as we can't sum up branches). However if instead of `sum' we put `count' as the aggregate operation then we obtain the analytic expression $(q, b, count)$ which {\em is} a well formed analytic expression, as `count' is an aggregate operation applicable on $b$-values. By the way, what this analytic expression returns is the number of branches per delivered quantity. \smallskip




\begin{figure}
{
\begin{center}
\includegraphics[width=350px,keepaspectratio]{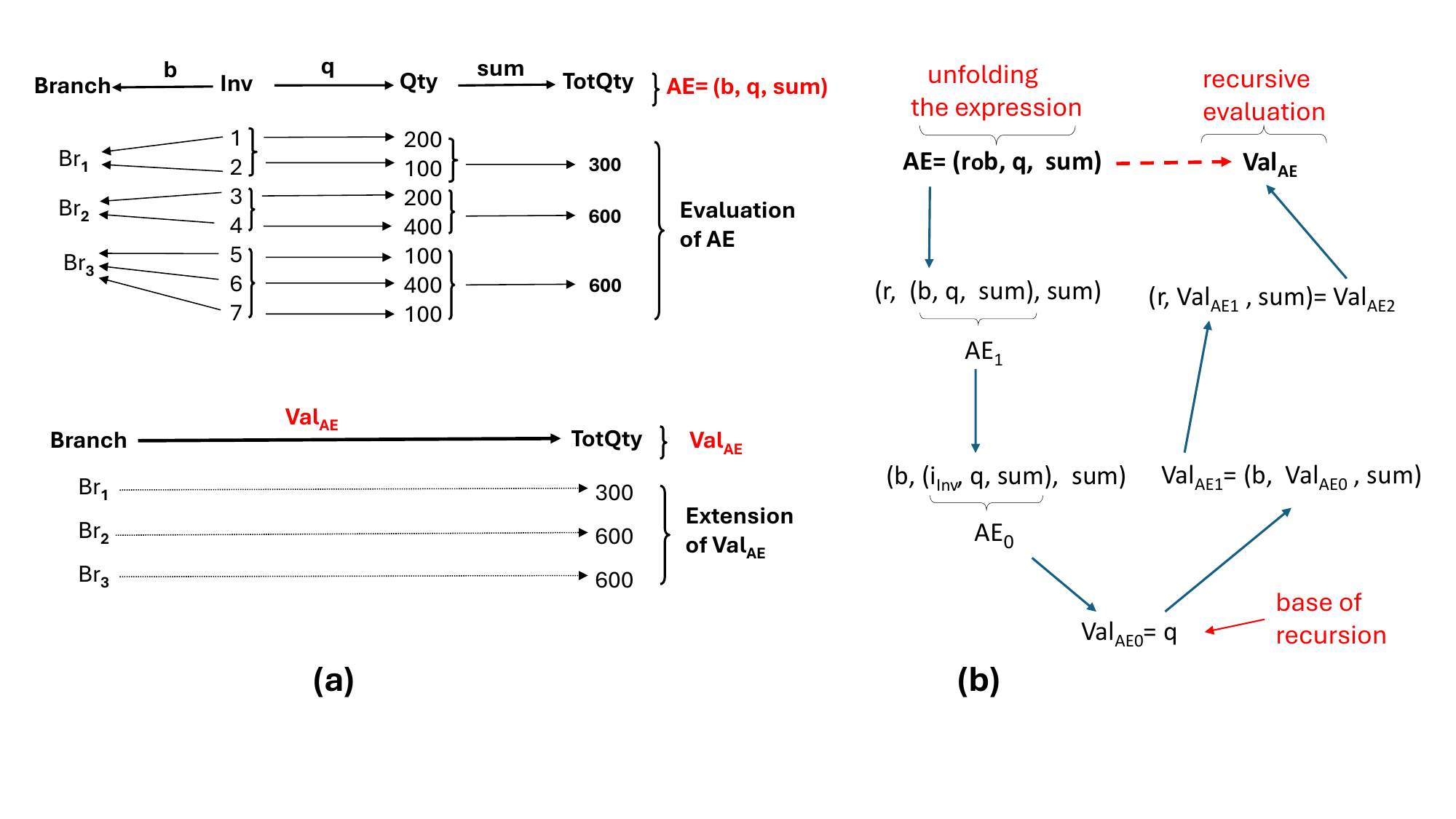}
\caption{(a) Evaluation of an analytic expression (b) Recursive evaluation of analytic expression\label{AE}}
\end{center}
}
\end{figure} 

\noindent Now, the concept of analytic expression as defined in the previous examples was first introduced in \cite{SpyratosS18}\cite{spyratos2023context} (although it is called `analytic query' in those works). In the present work, we extend the concept of analytic expression in several ways: 

\noindent (a) we introduce the use of traversal queries for performing grouping and measuring, (b) we define an analytic query to be either an analytic expression or the pairing of two or more analytic expressions with the same source, (c) we introduce a two-level rewriting for analytic queries, and (d) we introduce the concept of `mixed query' as the pairing of path expressions and analytic expressions all having the same source.

\subsection{The formal definition of analytic expression} In an analytic expression $(g, m, op)$, as defined informally in the previous example, the grouping and measuring functions are seen as abstract functions (or as edges of a database tree). On the other hand, as we have seen in the previous section, the answer to a traversal query is a function, therefore the (answers of) traversal queries can be used in the definition of an analytic query either as grouping or as measuring functions. This simple observation allows for a significant extension of the concept of analytic expression compared to the one introduced in \cite{SpyratosS18}\cite{spyratos2023context}. Indeed, by using traversal queries to define the grouping and measuring functions, a user has the possibility to define analytic expressions representing specific user needs. These observations motivate a semantically richer definition of analytic expression. 

\begin{definition}[Analytic expression]
    An {\em analytic expression} over a tree $\mathcal T$ is a triple $AE=(g, m, op)$, where $g$ and $m$ are traversal queries over $\mathcal T$  having the same source, and $op$ is an aggregate operation applicable on the target of $m$. 
    
    \noindent Given an instance $\delta_\mathcal T$ of $\mathcal T$, the {\em value} of $AE$, denoted $Val_{AE} (\delta_\mathcal T)$, or simply $Val_{AE}$ when $\delta_\mathcal T$ is understood, is the following function:

    \smallskip\noindent $Val_{AE}: range(g) \to target(op)$ 
    
    \noindent such that: $Val_{AE}(i)= op(m/(g^{-1}(i)))$, for all $i\in range(g)$ 

    \smallskip\noindent where $m/(g^{-1}(i))$ denotes the restriction of $m$ to the subset $g^{-1}(i)$ of its source 

     \smallskip\noindent The traversal queries $g$ and $m$ are called {\em grouping query} and {\em measuring query}, respectively, and they can be any traversal queries over $\mathcal T$.
\end{definition}

\noindent As in the case of a path expression, an analytic expression $AE=(g, m, op)$ has a source and a target. The source of $AE$ is the common source of the grouping and measuring queries, $g$ and m, and the target of $AE$ is the target of the aggregate operation $op$.  

\noindent Moreover, as in the case of a path expression, an analytic expression can be defined using  the `select/from/where/having/group-by' pattern of SQL as follows (see also section \ref{sec:TQ} , item 2): \smallskip

\noindent {\bf Select} $A_1, \ldots, A_m, $  {\bf op}($B_1, \ldots, B_n$) {\bf as} {\em Res}

\noindent {\bf From} $\langle tree-expression \rangle$

\noindent {\bf Where} $\langle conditions \rangle$

\noindent {\bf Group by}  $( A_1, \ldots, A_m)$  

\noindent  {\bf Having} $\langle conditions \rangle$ \smallskip

\noindent This statement should be interpreted as follows: 
\begin{enumerate}
\item The root (say $S$) of the tree defined by the tree expression in the {\em from} clause (call this tree  $\mathcal T$) is the common source of the grouping and measuring queries. 
\item The nodes $A_1, \dots, A_m$ in the {\em select} clause are the targets of the grouping query; $B_1, \dots, B_n$ are the targets of the measuring query; and $Res$ is a name (provided by the user) representing the target of the answer. 
\item The conditions in the {\em where} clause define (optionally) restrictions of the nodes $S, A_1, \dots, A_n, B_1, \dots, B_n$.  
\item The list of attributes in the {\em group-by} clause contains the targets of the grouping query (there may eventually be more than one list each containing a subset of the set $\{A_1, \ldots, A_m\}$ in case one wants to evaluate more than one analytic expression). 
\item   The conditions in the {\em having} clause define a restriction of  $Res$ (i.e. a restriction of the target of the value of the analytic expression). 
\end{enumerate}
\smallskip

\noindent To illustrate the definition of an analytic expression following the SQL pattern, consider the example  of  Figure \ref{AE}(a), where the evaluation of the analytic expression  $(b, q, sum)$ returns the totals by $Branch$, without any conditions on the grouping and measuring  queries. Its SQL-like definition is the following: \smallskip

\noindent {\bf Select} $Branch$  {\bf sum}($Qty$) {\bf as} {\em BrTot}

\noindent {\bf From} $\mathcal T_1$

\noindent {\bf Where} -

\noindent {\bf Group by}  $Branch$   

\noindent {\bf Having}  -  \smallskip

\noindent Here $BrTot$ on the {\em select} clause is the name of the result and the symbol `-' in the {\em Where} and {\em Having} clauses means no conditions are given. \smallskip

\noindent As yet another example consider the analytic expression  $(r\circ b, q, sum)$ over the database of Figure \ref{Fig-1}(a),  whose value is the total quantity by region without any conditions on the grouping and measuring  queries.  Its SQL-like definition is as follows:  

\smallskip
\noindent {\bf Select} $Region$  {\bf sum}($Qty$) {\bf as} {\em RegTot}

\noindent {\bf From} $\mathcal T_2\circ \mathcal T_1$ 

\noindent {\bf Where} -

\noindent {\bf Group by}  $Sup$  

\noindent {\bf Having} -  \smallskip

\noindent Here the symbol $RegTot$ on the {\em Select} clause is the name of the query result. \smallskip

\noindent An interesting remark regarding an analytic expression is that there are several possibilities of restriction. Here are a few examples:  

\begin{enumerate}
    \item We can restrict the grouping and/or the measuring query to a subset of the source and/or to a subset of a target of the query. 
    \item We can restrict the analytic expression itself. Indeed, as the value of an analytic expression $AE$ is a function, we can restrict it to a subset $D$ of its source (or to a subset of its target and then `pull' the restriction  back to its source). For example, consider the analytic expression $AE= (p, q, sum)$ over the tree $\mathcal T_1$ of  Figure~\ref{Fig-1}, which returns the totals by product. Its value is a function from $Prod$ to $Totals$, and if we define $D= \{x \in Prod / Val_{AE}(x) \leq 1000\}$ then the restricted value, $Val_{AE}/D$, will contain only products for which the total is less than or equal to 1000.   
    \item We can use any function which is applicable on the source of $Val_{AE}$ to define quite complex restrictions of $Val_{AE}$.  For example, consider the analytic expression: $AE= (s, c, count)$ over the tree $\mathcal T_3$ of  Figure~\ref{Fig-1}, which returns the number of categories supplied by each supplier. Its value is a function from the set $Sup$ of suppliers to a set $CountCat$ of integers, and if we define $D= \{x\in Sup / Val_{AE}(x) > Avg(Ans_Q(Sup))\}$ then the restricted value $Val(AE)/D$ will contain only suppliers supplying more categories than the average number of categories supplied by a supplier. 
\end{enumerate} 
We note that the kinds of restriction described in items 2 and 3 above are not possible to implement as relational algebra queries (although possible to implement as SQL group-by queries through the `Having' clause). So a relevant question here is: how do the possibilities of restriction offered by our model compare with those offered by the `Having' clause of SQL? However interesting, the answer to this question lies outside the scope of the present paper.

\noindent A few remarks are in order here explaining in what ways our definition of analytic expression generalizes significantly the one given in \cite{SpyratosS18}\cite{spyratos2023context}. First, as we shall see through examples shortly, by allowing grouping and measuring to be done by the results of traversal queries, we gain access to a much larger set of meaningful functions for grouping and measuring (`meaningful' in the sense that they are defined through queries expressing user needs). Second, as we shall see shortly, we can benefit from the combination of rewriting rules for path expressions with rewriting rules for analytic expressions so as to obtain a richer rewriting system for analytic queries. Third, by pairing two or more analytic expressions into a single analytic query, we need to partition the common source of the grouping queries only once and then apply the measuring query of each analytic expression separately. \smallskip

\noindent It is important to note that, in an analytic expression, the common source of the grouping and measuring queries can be any node - not necessarily the root of a database  tree. For example, consider the tree $\mathcal T_7$ of Figure~\ref{Fig-1} and the following two analytic expressions over $\mathcal T_7$: 

\noindent $(s, c, count)$ returning the number of product categories supplied by each supplier   

\noindent $(c, s, count)$ returning the number of suppliers supplying each product category  \smallskip 

\noindent The common source of the grouping and measuring queries in these analytic expressions is the node $Prod$ and not $Inv$. \smallskip 

\noindent Also note that, given a database tree $\mathcal T$, we can use the identity query $\iota_X$ of any node $X$ as the grouping or measuring query in an analytic expression over $\mathcal T$. Here is an example, referring to tree $\mathcal T_7$ of Figure~\ref{Fig-1}: 

$AE= (\iota_{Inv}, q, sum)$  \smallskip

\noindent During the evaluation of $AE$, in the grouping step, the function $\iota_{Inv}$ puts each invoice of $Inv$ in a  singleton block. Therefore summing up the values of $q$ in a group simply finds the value of $q$ on the single invoice in that group; then the measuring step simply returns this value of $q$. It follows that, for each invoice $i$, $Val_{AE}(i)$ returns $q(i)$. Therefore $Val_{AE}= q$. \smallskip

\noindent As another example over the same tree consider the following:

$AE= (q, \iota_{Inv}, count)$ 

\noindent During the evaluation of $AE$, in the grouping step, the traversal query $q$ groups together in a single group all invoices having the same delivered quantity; then, during the measuring step, the traversal query $\iota_{Inv}$ makes no changes in the group; and finally, during the aggregation step, the aggregate function $count$ counts the invoices in each group. It follows that, for each value $x$ of $Qty$, the value $Val_{AE}(x)$ is the number of invoices having $x$ as the delivered quantity (i.e. the answer is the number of invoices by quantity delivered). \smallskip

\subsection{Enriching the class of analytic expressions} In several occasions in practice one may need the answer to `summary queries' referring to the whole  period of data collection. Here are a few examples: 

what is the total quantity of products delivered (during the whole period of data collection)

what is the number of product categories present in the database 

what is the set (list) of suppliers present in the database

\noindent In order to be able to answer such queries,  we propose to `augment' the notion of tree as follows:
\begin{definition}
    Let $\mathcal T$ be a database tree. Then the augmented $\mathcal T$, denoted by $\mathcal T^{a}$, is the tree defined as follows: 
    \begin{itemize}
        \item for each node $X$ of $\mathcal T$ define an edge $\sigma_X: X\to \Sigma_X$, where $\Sigma_X$ is a new node such that $dom(\Sigma_X)$ is a singleton;
        \item the edge $\sigma_X$ is called the {\em constant edge} of $X$ 
    \end{itemize}
\end{definition}



\noindent Clearly, given a tree $\mathcal T$, we can use the constant edge $\sigma_X$ of each node $X$ of $\mathcal T^{a}$ in the same way as any other edge. In particular, we can use it in defining analytic expressions. Referring to the tree $\mathcal T_7$ of Figure~\ref{Fig-1}, here is an example: 

$AE= (\sigma_{Prod}, c, count)$

\noindent During the evaluation of $AE$, in the grouping step, the grouping query $\sigma_{Prod}$ groups all products in a single group; then the measuring query $c$ associates each product with a category; and finally, during the aggregation step, the aggregate operation $count$ finds the number of categories of all products delivered.

\smallskip
\noindent Here are two more examples involving the identity query and the constant query: 
\begin{itemize}
    \item $(\sigma_{Inv}, \iota_{Inv}, count)$ returns the number of all invoices of node $Inv$    \\
Note that the analytic expression $(\sigma_X, \iota_X, count)$ is typically used for finding the cardinality of a node $X$.\\
\item $(\sigma_{Inv}, q, sum)$ returns the total of all quantities delivered \\
Note that the analytic expression $(\sigma_X, m, op)$ is typically used for reducing the whole of $X$ under some measuring query $m$ and operation $op$.
\end{itemize} \smallskip

\noindent It is important to note that combinations of analytic expressions such as the above can be used in several special cases such as when we want to compute percentages. For example, in the tree $\mathcal T_1$ of Figure~\ref{Fig-1}, suppose we want to compute, for each branch, the percentage  of the total quantity of  products delivered. In this case, we use the following two analytic expressions: 
\begin{itemize}
    \item $AE_{Br}= (b, q, sum)$, returning the total quantity delivered by branch 
    \item $AE_{Tot}= (\sigma_{Inv}, q, sum)$, returning the total of all quantities delivered 
\end{itemize} 

\noindent By dividing now the value of $AE_{Br}$ on a specific branch $Br_i$ by the value returned by $AE_{Tot}$ we find the percentage for each branch $Br_i$ (and the set of percentages for all branches can be visualized for example in the form of a pie). \smallskip

\noindent Regarding the use of both $\iota_X$ and $\sigma_X$ in analytic expressions, we note that, for any edge $e: A \rightarrow B$, we have:

\begin{itemize}
  \item $\iota_B \circ e= e\circ \iota_A = e$ 
  \item $\sigma_B \circ e= \sigma_A$ 
\end{itemize}

\subsection{Evaluation of analytic expressions}
In principle, the evaluation of an analytic expression can be done in four different ways: 

\begin{enumerate}
\item Directly, as described in Figure \ref{AE}(a); 
\item By translating the analytic expression and evaluating it by some query engine; 
\item By sharing computations when evaluating a set of queries with the same grouping query; 
\item By rewriting the analytic expression so as to put it in some desired form. 
\end{enumerate} 

\noindent Direct evaluation consists essentially in programming the steps described in Figure \ref{AE}(a). Clearly, direct evaluation is inefficient for analyzing big data sets. \smallskip

\noindent A better option is translating the analytic expression and evaluating it by some query engine thus taking advantage of optimization techniques implemented in the query engine. An analytic expression as defined here can be translated as an SQL group-by query, when processing relational data \cite{SpyratosS18}; as a MapReduce job, when processing data residing in a file system \cite{SpyratosS18,ZervoudakisKSP21}; and as a SPARQL query when processing RDF data \cite{PapadakiST21}.  Moreover, the possibility of translating an analytic expression as a query in three different kinds of query engines makes it possible to use the tree database model as a `mediator' \cite{DBLP:journals/csur/Wiederhold95}. This means that a user formulates an analytic expression over the tree database; the expression is then translated to a query over the underlying query engine; and finally the user receives the answer `transparently', that is, as if the analytic expression were processed by the tree database. \smallskip

\noindent Sharing computations when evaluating two or more analytic expressions is a third option. Consider the following set of analytic expressions:

$\mathcal {AE}= \{(g, m_1, op_1), \ldots, (g, m_n, op_n)\}$ 

\noindent The expressions in this set share the same  grouping query $g$. Clearly, in this case, the result of the grouping step will be the same for each analytic expression $AE_i, i= 1, \ldots, n$. Therefore the set of analytic expressions can be rewritten using the `optimized' form $\mathcal {AE}'= \{(g, \{(m_1, op_1), \ldots, (m_n op_n))\}$, meaning that the grouping step is performed only once and then the measuring and aggregation steps for each analytic expression is performed to obtain the results of the analytic expressions in $\mathcal {AE}$. Moreover, if the analytic expressions in the set share both, the grouping {\em and} the measuring query then the set $\mathcal {AE}$ can be rewritten using the `optimized' form: 

$\mathcal {AE}'=(g, m, \{op_1, \ldots, op_n\})$ 

\noindent This time, the grouping {\em and} measuring steps are performed only once and then the aggregation step for each analytic expression is performed to obtain the results of the analytic expressions in $\mathcal {AE}$. \smallskip

\noindent To illustrate these optimizations, refer to the tree $\mathcal T_7$ of Figure~\ref{Fig-1} and consider the following set of two analytic expressions:

$\mathcal {AE}= \{(b, q, sum), (b, s\circ p, count)\}$.

\noindent The first asks for the total quantity delivered by branch and the second asks for the number of suppliers by branch. Then the optimized form for the set $\mathcal {AE}$ is the following:  

$\mathcal {AE}'= (b, \{(q, sum), (s\circ p, count)\})$.

\noindent Next, consider the following set of two analytic expressions: 

$\mathcal {AE}= \{(b, q, sum), (b, q, avg) \}$. 

\noindent The first asks for the total quantity delivered by branch and the second for the average quantity delivered by branch. Then the optimized form is:

$\mathcal {AE}'= (b, q, \{sum, avg\})$.

\noindent \subsection {Rewriting an Analytic Expression}\label{RAE}

\noindent In this section we discuss the last option for evaluating an analytic expression that is by rewriting the expression. Rewriting an analytic expression $AE= (g, m,op)$ can be done at two orthogonal levels: (a) rewriting the grouping query $g$ and/or the measuring query $m$, and (b) rewriting the analytic expression itself in terms of other analytic expressions. Rewriting the grouping or the measuring query can be done as outlined in section \ref{sec:TQ}. Rewriting the analytic expression in terms of other analytic expressions depends on the form of the grouping and measuring queries as well as on the kind of aggregate operation used in the analytic expression. More precisely, if the aggregate operation is associative \cite{LaurentS11}, then the following is the basic rule for rewriting an analytic expression \cite{SpyratosS18}: \smallskip

{\bf {\em Composition Rule}} :  $(g' \circ g, m, op)= (g', (g, m, op), op)$ \smallskip

\noindent What this rule says is that if the grouping query is the composition of two other queries then the analytic expression can be rewritten as a nested analytic expression. To understand this nesting, consider the following analytic expression over the tree $\mathcal T_7$ of Figure \ref{Fig-1}: \smallskip

$AE= (r \circ b, q, sum)$, asking for the totals by region \smallskip

\noindent Here is why the composition rule works: the edge $r$ tells us in which region each branch is located, therefore, if we have the totals for each branch in a region then we can sum them up to find the total for that region. Now, to find the totals by branch we can use the analytic expression $AE'= (b, q, sum)$; and as the answer to $AE'$ has the same source as $r$ we can use the analytic expression  $AE''= (r, Val_{AE'}, sum)$ to compute the totals by region. In other words, the expression $AE$ for computing the totals by region can be rewritten as $AE= (r, Val_{AE'}, sum)$. By replacing now $Val_{AE'}$ by $(b, m, sum)$ we obtain that  $AE= (r, (b, q, sum), q)$, which should be read as follows: first evaluate the inner or nested expression $(b, m, sum)$ to obtain the totals by branch; and then use its value as the measure to evaluate the outer expression thus obtaining the totals by region.\smallskip 

\noindent It is important to note that the composition rule works only if the aggregate operation $op$ is associative \cite{LaurentS11}. Most common operations such as `sum', `min', `max' etc. are associative, for example, $sum(1, 2, 3, 4, 5)=sum(sum(1, 2), sum(3, 4, 5))$. In contrast, `avg' is not associative, for example: 

$avg(1, 2, 3, 4, 5) \neq avg(avg(1, 2), avg(3, 4, 5))$)

\noindent Although there are `corrective' algorithms allowing the use of many non-associative operations such as average, median, etc. (see for example~ \cite{ZervoudakisKSP21}\cite{LaurentS11}), we shall not pursue this subject any further here. Rather, in order to simplify the discussion, we shall assume that all aggregate operations are associative.\smallskip

\noindent The composition rule allows for the recursive evaluation of analytic expressions as shown in Figure \ref{Fig-2}(b), where the expression $AE= (r \circ b, q, sum)$, asks for totals by region. We first unfold $AE$ using the composition rule until we reach the base expression $AE_0$; then we evaluate the base expression and use its value to compute the value of the expression $AE_1= (b, q, sum)$; and finally we use the value of $AE_1$ to compute the value of $AE$. Note how the identity query $\iota_{Inv}$ is indispensable in defining the base expression $AE_0$, based on the equality: $b= b \circ \iota_{Inv}$. Indeed, it follows from this equality that: 

\noindent $AE_1= (b, q, sum)= (b \circ \iota_{Inv}, q, sum)= (b, (\iota_{Inv}, q, sum), sum)$ 

\noindent Therefore we have:  $AE_0= (\iota_{Inv}, q, sum)$ and $Val(AE_0)= q$ 

\smallskip
\noindent Using the composition rule  we can sometimes avoid making (relational) joins during the evaluation of an analytic expression. To see how, consider again the expression $AE= (r \circ b, q, sum)$ and suppose that the answers of $r$ and $b$ are stored on separate tables, say $R$ and $B$, respectively. The direct evaluation of $AE$ requires a join between $R$ and $B$ in order to find the answer to the grouping query, whereas if we use the rewritten form $AE'= (r, (b, q, sum), sum)$, no join is needed since the grouping functions for the inner and outer expressions are stored separately, in $R$ and $B$, respectively. Another advantage of using the composition rule is the possibility of incremental evaluation of analytic expressions when processing big data sets \cite{SpyratosS18}\cite{ZervoudakisKSP21}. \smallskip

\noindent Now, the composition rule says how to rewrite an analytic expression in terms of other analytic expressions, no matter whether we rewrite or not its grouping and measuring queries. Let's now see an example of rewriting an analytic query at both levels: first rewriting its grouping and/or measuring query to obtain some desired (equivalent) form; then rewriting the resulting analytic expression using the composition rule. For example, consider the following analytic expression:

$((s \circ p) \wedge (c \circ p), q, sum)$  asking for the totals by supplier and category
 
\noindent Here, by applying lemma \ref{BasicLemma} to the grouping query $(s \circ p) \wedge (c \circ p)$ we can rewrite it as $(s \wedge c) \circ p$, so we can now rewrite the analytic expression as follows:

$((s \wedge c) \circ p, q, sum)$

\noindent Next, we apply the composition rule to obtain the final, rewritten expression:

$((s \wedge c), (p, q, sum), sum)$

\smallskip
\noindent As this example shows, using the composition rule together with rewriting of the grouping and/or measuring query, we can derive a variety of rewriting rules. A most notable example is what we call the `pairing rule'. 

\noindent Consider an analytic expression of the form $AE= (g_1 \wedge g_2, m, op)$, where the grouping query is the pairing of two queries, $g_1$ and $g_2$. Suppose now that we want to evaluate the expression $AE_1=(g_1, m, op)$. It follows from Lemma  \ref{BasicLemma} that $g_1= \pi_{target(g_1)} \circ (g_1 \wedge g_2)$, therefore $AE_1$ can be rewritten as follows:  ~$AE_1=(\pi_{target(g_1)} \circ (g_1 \wedge g_2), m, op)$ 

\noindent Applying now the composition rule, we can further rewrite $AE_1$ as follows:

$AE_1= (\pi_{target(g_1)} \circ (g_1 \wedge g_2), (g_1 \wedge g_2, m, op), op)$

\noindent Following the same reasoning for the expression $AE_2= (g_2, m, op)$ we have:

$AE_2= (g_2, m, op)= (\pi_{target(g_2)} \circ (g_1 \wedge g_2), (g_1 \wedge g_2, m, op), op)$ 

\noindent Hence the following rule for rewriting an expression $AE_i= (g_i, m, op)$ in terms of the expression $AE= (g_1 \wedge \ldots \wedge g_n, m, op)$, $n \geq 2$ : \\

\noindent {\bf {\em Pairing Rule}}: \smallskip

$AE_i= (g_i, m, op)= (\pi_{target(g_i)} \circ (g_1 \wedge \ldots  \wedge g_n), (g_1 \wedge \ldots  \wedge g_n, m, op), op)$ \\


\noindent To see how the pairing rule works, refer to the tree $\mathcal T_7$ of Figure \ref{Fig-1} and suppose that we have computed the totals by branch and product. Then we can re-use this result to compute the totals by branch and the totals by product in terms of the totals by branch and product. Indeed, applying the pairing rule we have:

\noindent $(b, q, sum)= (\pi_{Branch} \circ (b \wedge p), (b \wedge p, q, sum), sum)$ 

\noindent $(p, q, sum)= (\pi_{Product}\circ (b \wedge p), (b \wedge p, q, sum), sum)$.

\subsection{Analytic queries}
In the previous subsection we defined and discussed analytic expressions and their evaluation. In this subsection we define analytic queries by proceeding in a similar manner as for traversal queries that is we define an analytic query to be either a single analytic expression or the pairing of two or more analytic expressions with the same source. 
As we already mentioned, in our model, these two types of queries are defined using the {\em same} formal framework, namely the functional algebra. In contrast, in the relational model, it is not possible to define analytic queries in the relational algebra; they are defined outside the relational algebra, in the form of SQL group-by queries.

\begin{definition}[Analytic query]
\noindent Let $\mathcal D$ be a tree database and $\mathcal T$ a tree in $\mathcal D$. An {\em analytic query} $Q$ over $\mathcal T$ is defined to be either a single analytic expression $AE$ over $\mathcal T$ or the pairing of two or more analytic expressions with the same source $S$ and distinct targets that is: 

$Q= AE_1\wedge AE_2\wedge \ldots \wedge AE_n$, $n\ge 1$




\noindent Given an instance $\delta_\mathcal T$, the {\em answer} to $Q$ is denoted by $Ans_Q(\delta_\mathcal T)$ and it 
is defined as the pairing of the values of the analytic expressions $AE_1, AE_2, \ldots AE_n$ that is:

$Ans_Q(\delta_\mathcal T)= Val_{AE_1}(\delta_\mathcal T)\wedge \ldots \wedge Val_{AE_n}(\delta_\mathcal T)$

\noindent We shall also denote this answer by $Ans_Q$ whenever $\delta_\mathcal T$ is understood from context.
\end{definition}
\subsection{Mixed queries}\label{subsec:MQ} 
We have seen that a traversal query over a tree $\mathcal T$ is either a single path expression over $\mathcal T$ or the pairing of two or more path expressions over $\mathcal T$ having the same source; and that an analytic query is either a single analytic expression over $\mathcal T$ or the pairing of two or more analytic expression over $\mathcal T$ having the same source. Hence we can unify the two types of expression over $\mathcal T$ in the following definition of `mixed query' over $\mathcal T$.


\begin{definition}[Mixed query]
\noindent Let $\mathcal D$ be a tree database and let $\mathcal T$ be a tree in $\mathcal D$. A {\em mixed query} $Q$ over $\mathcal T$ (or simply {\em query} over $\mathcal T$) is defined to be the pairing of two or more expressions over $\mathcal T$ having the same source and distinct targets: 

$Q= E_1\wedge \ldots \wedge E_n$, $n\ge 1$ 

\noindent where each $E_i$ is either a path expression or an analytic expression over $\mathcal T$.



\noindent Given an instance $\delta_{\mathcal T}$ of $\mathcal T$, let $f_1, \ldots, f_n$ be the functions to which $E_1, \ldots, E_n$ evaluate in $\delta_{\mathcal T}$, respectively. Then the {\em answer} to $Q$ is denoted by $Ans_Q(\mathcal T)$ (or simply by $Ans_Q$ whenever $\mathcal T$ is understood) and it is defined as follows: 

$Ans_Q= f_1\wedge \ldots\wedge f_n$

\end{definition} 
For example, referring to Figure \ref{Fig-1}, the following is a mixed query over the tree $\mathcal T_7$:

\smallskip
$Q= r\wedge (b, q, sum)$

\smallskip
\noindent This query returns, for each branch $x\in Branch$, the pair $(r(x), BrTot(x))$ that is the region to which the branch $x$ is located and the total quantity, $BrTot(x)$, delivered to  branch $x$.

\smallskip
\noindent Note that, as in the case of traversal queries, a mixed query $Q$ induces a relation schema $R_Q$, and for every instance $\delta_{\mathcal T}$, a relation $r_Q$ over $R_Q$. In our previous example $R_Q$ has three attributes, $Region$, $Branch$ and $BrTot)$ and two functional dependencies, $Branch\to Region$ and $Branch\to BrTot$. Moreover, for every instance of $\mathcal T$, the relation $r_Q$ is the answer to $Q$ in that instance and satisfies the two dependencies. 

\noindent Also note that in a mixed query we can apply the evaluation techniques that we have seen earlier by grouping together path expressions or by grouping together analytic expressions in the query. For example, consider the following mixed query over the tree $\mathcal T_7$ of Figure \ref{Fig-1}, which consists of one path expression and two analytic expressions: \smallskip





$Q= r\wedge (b, q, sum)\wedge (b, q, avg)$ \smallskip

\noindent This query asks, for each branch, the region in which the branch is located together with the total and the average for the branch. 

\noindent If we group together the analytic expressions then we can replace their pairing by the following expression: $(b, q, \{sum, avg\})$; and then we can replace the query $Q$ with the following query (which is less expensive to evaluate):\smallskip 

$Q'= r\wedge (b, q, \{sum, avg\})$ \smallskip

\noindent We end this section by noting that (a) traversal queries and analytic queries have similar syntax although they have different semantics and (b) mixed queries offer the possibility to use combinations of traversal and analytic queries in a principled way. In a sense, we can say that the query language of a tree database is actually the set of all mixed queries. 

\section{Applications}\label{sec:Apps} 

In this section we discuss the expressive power of our model by outlining three possible applications, namely (a) how attribute inheritance can be defined seamlessly in our model using  one-to-one edges and tree composition (b) how one can define a consistent relational database as a set of mixed queries over a tree database - with the tree database playing the role of an underlying semantic layer; and (c) how a tree database can be used as a user-friendly interface for accessing and analyzing relational data. 

\subsection{Inheritance}\label{subsec:Inh} 

One of the most intriguing - and at the same time most problematic - notions in object-oriented programming is inheritance \cite{DBLP:journals/csur/Taivalsaari96}. Inheritance is commonly regarded as the feature that distinguishes object-oriented programming from other modern programming paradigms, but researchers rarely agree on its meaning and usage. Yet inheritance is often hailed as a solution to many problems hampering software development, and moreover, many of the alleged benefits of object-oriented programming, such as improved conceptual modeling and re-usability, are largely credited to it. Inheritance is the fundamental reuse mechanism in object-oriented programming languages.\smallskip

\noindent In our model, inheritance is the process that allows an entity (i.e. the root of a tree)  to acquire or {\em inherit} the attributes of another entity (i.e. of the root of another tree) - in much the same way as in  traditional databases, where inheritance allows a table to acquire the properties of another table. The tree that inherits is called the {\em child tree} and the tree whose attributes are inherited is called the {\em parent tree}. 

\begin{definition}[Inheritance]
Let $\mathcal D$ be a tree database and $\mathcal T, \mathcal T'$ two disjoint trees of $\mathcal D$ (i.e. no node in common) with roots $\rho$ and $\rho'$, respectively. We say that $\rho$ {\em inherits} from $\rho'$ if the following hold: 
\begin{enumerate}
    \item dom($\rho$) = dom($\rho'$) 
    \item $\delta_{\mathcal T}(\rho)\subseteq \delta_{\mathcal T'}(\rho')$, in every database instance 
    \item there is an edge $Isa: \rho\to \rho'$ such that, in every database instance, we have:  

$Isa(x)= x$ for all $x\in \rho$ (i.e. $Isa$ is a 1-1 edge) 
\end{enumerate} 
\end{definition}

\noindent A word of explanation is necessary here in order to understand how $\rho$ can {\em inherit} from $\rho'$ by adding an Isa edge from $\rho$ to $\rho'$.  First note that, as $\mathcal T, \mathcal T'$ are disjoint trees, $\mathcal T$ is not composable with $\mathcal T'$. However, 
by adding an $Isa$ edge from $\rho$ to $\rho'$, we can view $\mathcal T$, together with the added $Isa$ edge, as a single tree, say  $\mathcal T^+$. In this way, $\rho'$ becomes a leaf of $\mathcal T^+$ while being the root of   $\mathcal T'$. Therefore $\mathcal T^+$ becomes composable with $\mathcal T'$. It follows that, in the tree $\mathcal T' \circ \mathcal T^+$, there is a unique path from $\rho$ to every attribute $A$ of $\mathcal T'$. Therefore $\rho$ inherits every attribute $A$ of $\rho'$ through the unique path from $\rho$ to $A$. 

\noindent We note in passing that performing the composition $\mathcal T' \circ \mathcal T^+$ requires restriction propagation only in $\mathcal T'$ (but not in $\mathcal T$, because of item 2 in the above definition). On the other hand, the database must satisfy the constraint expressed by item 3 in the above definition in order to support the inheritance represented by the $Isa$ edge. Finally note that, as a consequence of item 3 above, the composition of $Isa$ edges is also an $Isa$ edge.  \smallskip

\noindent To illustrate how inheritance can be supported by our model, suppose that we have two trees in the database, the $Car$ tree and the $Boat$ tree each with their own attributes, as shown in Figure \ref{Fig-3}(a). The $Car$ attributes are $Country, Manuf, Price$, and the $Boat$ attributes are $Masts, Manuf, Price$ (`Masts' meaning `number of masts' in the boat). As the two entities $Car$ and $Boat$ have the attributes $Manuf$ and $Price$ in common, the idea is to (a) `factor out' (i.e. remove) the common attributes $Manuf$ and $Price$ from the $Car$ tree and from the $Boat$ tree, (b) create a third entity, $Vehicle$, having $Manuf$ and $Price$ as its attributes and (c) place two $Isa$ edges, $Isa_{CV}: Car\to Vehicle$ and $Isa_{BV}: Boat\to Vehicle$ so as $Car$ and $Boat$ each inherit the attributes $Manuf$ and $Price$ of $Vehicle$ (in the way described above). This process is depicted in Figure \ref{Fig-3}(b).

\noindent Note that a tree can inherit from one or more disjoint trees and can be inherited by one or more disjoint trees (using tree composition). \smallskip

 \begin{figure}
{
\begin{center}
\includegraphics[width=350px,keepaspectratio]{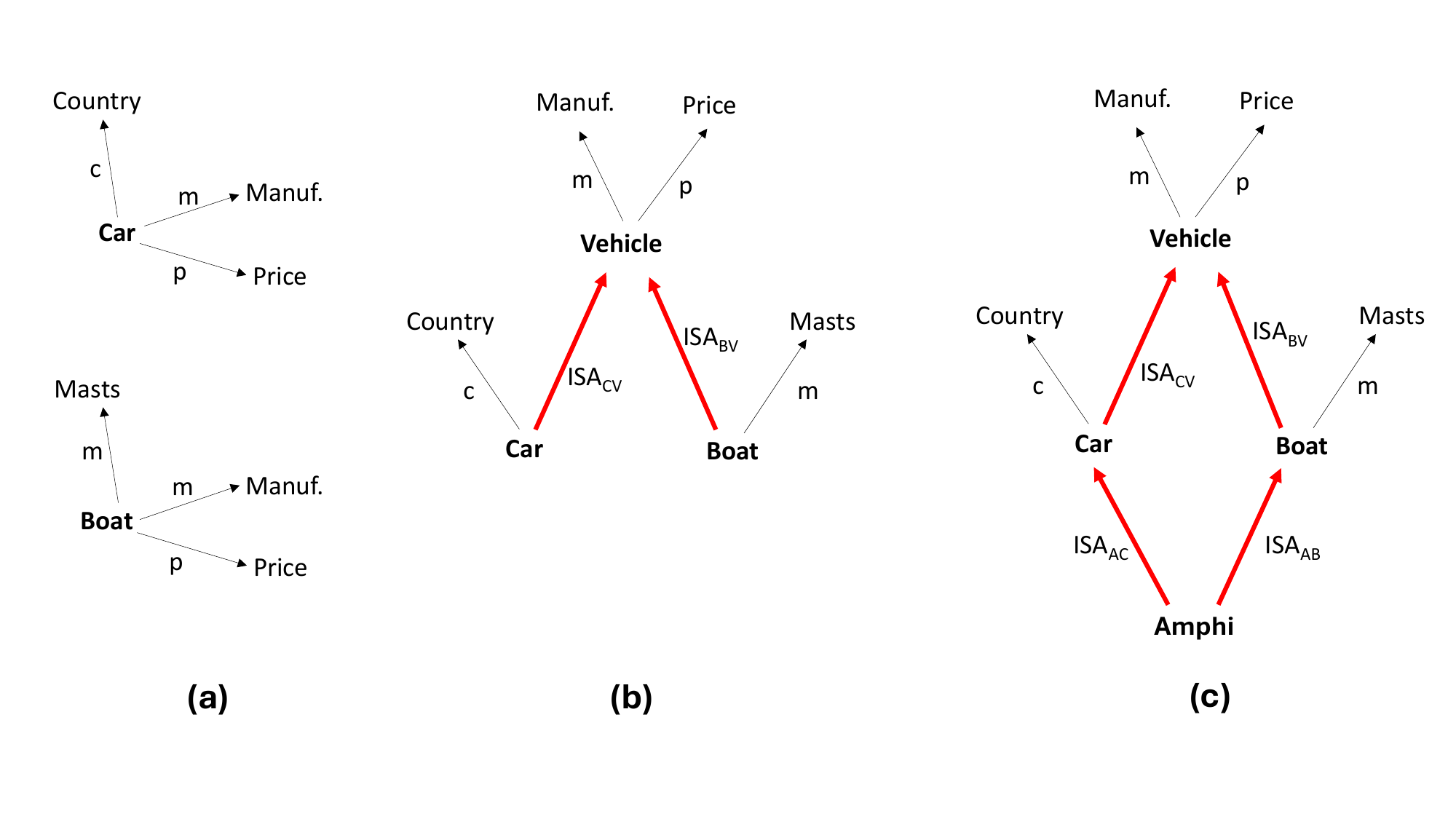}
\caption{(a) The Car tree and the Boat tree (b) Inheritance by Car and Boat from the Vehicle tree (c) Multiple inheritance by the Amphi tree and the diamond problem \label{Fig-3}}
\end{center}
}
\end{figure}

\noindent Going one step further, suppose now that we want to create an entity $Amphi$ (short for `Amphibious') that inherits the attributes of $Car$ and $Boat$. To do so we can create a new tree, the $Amphi$ tree (with only one node, $Amphi$) and add to it two $Isa$ edges:  $Isa_{AC}: Amphi\to Car$ and $Isa_{AB}: Amphi\to Boat$, as shown in Figure \ref{Fig-3}(c). Then $Amphi$ inherits the attributes $Country$ from $Car$ and $Masts$ from $Boat$ (this kind of inheritance from two or more other trees is called `multiple inheritance'). Now however, we have a problem as $Amphi$ inherits $Make$ and $Price$ by composing the $Amphi$ tree with the $Vehicle$ tree (based on the $Isa$ edge $Isa_{CV}\circ Isa_{AC}$), and at the same time $Amphi$ inherits the same attributes $Make$ and $Price$ by composing the $Amphi$ tree with the $Boat$ tree (based on the $Isa$ edge $Isa_{BV}\circ Isa_{AB}$). However, nothing guarantees that these two ways of inheriting will return the same values for $Make$ and $Price$! This problem, known as the {\em Diamond Problem} drew much controversy in the object-oriented programming community as to whether we should use multiple inheritance \cite{DBLP:journals/csys/Waldo91}. In our context, we can envisage two solutions to the problem: \smallskip

\noindent - maintain the constraint $Isa_{CV}\circ Isa_{AC}= Isa_{AB}\circ Isa_{BV}$ during database updates

\noindent - or require that the set of $Isa$ edges in the database form one or more trees 

(i.e. no parallel paths formed by $Isa$ edges - as is the case in Figure \ref{Fig-3}) \smallskip

\noindent The main advantage of using inheritance in a tree database is the same as in a classical database that is: inheritance allows a newly created tree to acquire the information contained in an existing tree, thus lowering redundancy by creating a hierarchical structure of trees expressed by the $Isa$ edges present in the database.  
Inheritance hierarchies are crucial to building a structured and well-organized tree database. However, in such a database, the $Isa$ edges should form one or more trees otherwise we may run into problems (such as the diamond problem). \smallskip

\noindent As a last remark, in presence of multiple inheritance, as in the example of Figure~\ref{Fig-3} we can (a) express semantic constraints and (b) use a set-theoretic language to extract additional, useful information from the database. For example, the following set-theoretic expressions might be imposed as a constraint in our example database:

$\delta_\mathcal T(Vehicle)= Isa_{CV}(Car)\cup Isa_{AB}(Boat)$ 

\noindent expressing the fact that there can't be vehicles in the database other than cars and boats. Clearly, as usual, such a constraint will have to be maintained during database updates. 

\noindent As an example of semantic query involving $Isa$ edges consider the following: are there vehicles in the database which are neither cars nor boats? This query can be answered by computing the following set-theoretic expression: \smallskip

$E= \delta_\mathcal T(Vehicle) \subseteq (Isa_{CV}(Car)\cup Isa_{AB}(Boat))$  \smallskip

\noindent If the value of $E$ is `true' then there is no vehicle in the database which is neither a car nor a boat; and if the answer is `false' then there is vehicle in the database which is neither a car nor a boat. \smallskip




\noindent As a last remark, it is clear that the subject of inheritance in tree databases offers a rich ground for research, and in fact this is a main subject of our current work (see also the following Section \ref{sec: Conclusions}). 

\subsection{Defining relations over a tree database} 


We have seen that, in a tree database $\mathcal D$, a mixed query over a tree $\mathcal T$ (hereafter called simply `query') is of the form $Q= E_1\wedge E_2\wedge \ldots \wedge E_n$, where each $E_i$ denotes either a path expression or an analytic expression from a common source node $S$ to a target node $A_i$, $i= 1, 2, \ldots, n$. We have also seen that, given an instance of $\mathcal T$, each $E_i$ evaluates to a function $f_i:S\to A_i$, so the answer to $Q$ is the pairing $f_1\wedge \ldots\wedge f_n$ (i.e. a function from $S$ to $A_1\times \ldots \times A_n$). Finally, we have also seen that each query $Q$ induces a relation schema $R_Q(S, A_1, A_2, \ldots, A_n)$; and that, for every instance $\delta_\mathcal T$ of $\mathcal T$, the answer of $Q$ in $\delta_\mathcal T$ induces a relation $r_Q$ over $R_Q$ which satisfies the functional dependencies $f_1, f_2, \ldots, f_n$ (see Section \ref{subsec:MQ}). 

\noindent In other words, each query $Q$ over $\mathcal D$ induces a relation schema $R_Q$ and a relation $r_Q$ over $R_Q$ which is consistent with the functions defining the answer of $Q$. 
For example, consider the following queries: \smallskip

\noindent $Q_1= s\wedge c$, a traversal query returning, the supplier and category of each product

\noindent $Q_2= (p, q, sum)$, an analytic query returning the total quantity by product (call it $ProdTot$). \smallskip

\noindent These queries induce each a relation schema as follows: \smallskip

\noindent $R_{Q_1}(Prod, Sup, Cat)$ with dependencies: $Prod\to Sup$ and $Prod\to Cat$ 

\noindent $R_{Q_2}(Prod, ProdTot)$ with dependency $Prod\to ProdTot$ \smallskip

\noindent The above two relation schemas define a relational database schema $RDB= \{R_{Q_1}, R_{Q_2}\}$; and, for each instance of $\mathcal D$, the answers of $R_{Q_1}$ and $R_{Q_2}$ in that instance induce the relations $r_{Q_1}$ and $r_{Q_2}$, respectively. Moreover, $r_{Q_1}$ and $r_{Q_2}$ are consistent each with respect to their dependencies. \smallskip

 \noindent The important thing to note here is that the values of each non-key attribute $A_i$ of $r_Q$ are computed by evaluating the expression $E_i$ of $Q$ of which $A_i$ is the target. Therefore the expression $E_i$ provides the {\em semantics} of attribute $A_i$ of $R_Q$, $i= 1, \ldots, m$. As a consequence, when defining a relational database over a tree database we have a significant semantic gain: the relation schema $R_Q(S, A_1, \ldots, A_n)$ comes with a clear semantics for its attributes, this semantics being the expression defining each attribute. In contrast, in a traditional database, the attributes of a relation schema have only the intuitive semantics conveyed by their names. Therefore when a relational database is defined as a set of queries over a tree database, the tree database works as the underlying `semantic layer' of the relational database thus defined. \smallskip

\noindent We end this subsection by two remarks concerning the use of relational databases defined over of a tree database:
\begin{enumerate}
    \item The user of a relational database defined over a tree database actually `sees' a set of relation schemas (i.e. a set of tables), exactly as when interacting with a traditional relational database. Here, however, in contrast to a traditional relational database, each table is associated with clear semantics that the user can inspect assuming that (a) the expression computing the values of each attribute $A_i$ is associated with $A_i$ (but stays invisible to the user) and (b) all attributes of a table are made clickable so as the  interested user can click an attribute of a table and `see' the expression used to compute the values of that attribute. This expression says how the functions of the underlying tree database are combined in order to compute the values of that attribute. \smallskip
    \item It is natural that the user of relational database defined by a set of queries over a tree database should be able to ask relational algebra queries, as usual (or SQL queries in general). Therefore the system will have to translate each user query as a query over the underlying tree database and return the answer of the user query in the form of a relation. Designing and implementing such translation algorithms is part of our current research. 
\end{enumerate} 

\subsection{A tree database as an interface}\label{subsec:Interface} 

The idea to use a tree database as an interface to a relational database requires to answer a larger question, namely: under what conditions can we transform a relational database to an `equivalent' tree database? To answer this question, we proceed in two steps as follows: 
\begin{itemize}
    \item {\em\bf Step 1}: Transform the relational database schema $\mathcal S$ to a set of trees.
    \item {\em\bf Step 2}: For each (consistent) instance of a relational database over $\mathcal S$ define 
    
an `equivalent' instance of the tree database of the previous step. 
\end{itemize}
Implementing the above  steps is a long term project and here we only outline some first ideas and results. First, we recall a few basic facts from relational database theory \cite{Ullman83}. \\

\noindent A {\em minimal cover} of a set $\mathcal F$ of (functional) dependencies is a set $\mathcal F'$ of dependencies which is equivalent to $\mathcal F$ and such that:
\begin{enumerate}
    \item The right-hand side of each dependency in $\mathcal F'$ is a singleton.

    \item The left-hand side of each dependency in $\mathcal F'$ is minimal that is there is no redundant (or extraneous) attribute.

    \item $\mathcal F'$ is minimal that is it contains no redundant dependency $f$ 
    
    ($f$ is redundant if implied by the set of dependencies $\mathcal F'\setminus \{f\}$). 
\end{enumerate} 
\smallskip
\noindent We also recall that $\mathcal F$ is equivalent to $\mathcal F'$ if $\mathcal F$ implies every dependency of $\mathcal F'$ and conversely, $\mathcal F'$ implies every dependency of $\mathcal F$. It can be shown that the minimal cover of a set of dependencies always exists but is not necessarily unique \cite{Ullman83}. \smallskip

\noindent Now, to implement the first step above, let $(R(\mathcal A), \mathcal F)$ be a relation schema, where $R$ is the name of the schema, $\mathcal A$ the set of attributes and $\mathcal F$ the set of dependencies that each relation over $\mathcal A$ must satisfy in order to be consistent. Moreover, we assume that $\mathcal F$ is `minimized' that is $\mathcal F$ is equal to one of its minimal covers. 
Now, a relation $r$ over $R$ is consistent only if the following hold: 
\begin{itemize}
    \item For every dependency $f: X\to A$ in $\mathcal F$, the projection $\pi_{XA}(r)$ is a function from $\pi_{X}(r)$ to $\pi_{A}(r)$ 
    \item There are no two different dependency paths with the same source and the same target (because then no relation over $R$ with more than one tuple can be consistent - see previous item)
    \item There cannot be a cycle of functional dependencies over an attribute $A$ of $\mathcal A$ (or over a set of attributes of $\mathcal A$) other than the identity cycle on $A$. 
    
    Indeed, attribute domains are abstract sets of values that is the values of an attribute domain are pairwise independent; therefore, if such a cycle existed on $A$ and we denoted by $\mathcal C$ the composition of its functions, then for every $a\in dom(A)$ we must have $\mathcal C(a)= a$ otherwise, if $\mathcal C(a)= a'$ and $a'\neq a$ this would mean that $a'$ depends on $a$, a contradiction. 
\end{itemize} 
It follows that each key $K$ of $\mathcal F$, together with all dependency paths having $K$ as source form a tree with root $K$. As a consequence, the set of trees induced by all relation schemas in $RDB$ forms a tree database, say $TDB$. \smallskip

\noindent Now, to implement the second step above, given a relation schema $(R(\mathcal A), \mathcal F)$ in $RDB$, an instance $r$ of $R$ and a tree $\mathcal T$ of $TDB$ induced by $\mathcal F$, we can define the instance of each node $N$ and each edge $e:X\to Y$ of $\mathcal T$ as follows: $\delta_\mathcal T(N)= \pi_N(r)$ and $\delta_\mathcal T(e)= \pi_{XY}(r)$ (which is a total function from $X$ to $Y$ since $r$ satisfies the dependency $X\to Y$). The tree database $TDB$ defined in this way can then act as an interface to $RDB$. Moreover, this interface can be tailored to the needs of specific users or groups of users by creating tree views of $TDB$ as defined earlier (see Section \ref{subsec:Views}).

\noindent Using the view mechanism the interface administrator can define a set of tree views tailored to the needs of a user or a group of users. 

\noindent We have already done some preliminary work along these lines as reported in \cite{VitsaxakiNMS22}.

\section{Concluding remarks and perspectives}\label{sec: Conclusions} 

We have seen a novel database model in which the database is seen as a set of trees, where the nodes of each tree represent the data sets of an application and the edges represent total functions between datasets. The query language of our model offers two kinds of queries, traversal queries and analytic queries, whose definitions are based on a set of elementary operations on functions that we called, collectively, the functional algebra. 
A traversal query over a tree is the pairing of path expressions with the same source, 
whereas an analytic query over a tree is the pairing of analytic expressions with the same source.  
Whether traversal or analytic, a query is always defined over a tree which can be either a tree in the database or a tree defined from other trees using a set of operations on trees that we called, collectively, the tree algebra. We have also seen that both traversal and analytic queries can be expressed using SQL-like statements.   
We have also outlined how our model can support inheritance; how a consistent relational database can be defined on top of a tree database (with the tree database providing the underlying semantics for attributes and relations); and how a tree database can be used as a user friendly interface for accessing and analyzing relational data. \smallskip

\noindent This work follows our previous work on graph database modeling in which the whole database is represented as a single graph, namely as a labeled directed acyclic graph with a single root called a `context' \cite{spyratos2023context}. However, we believe that, using a set of trees together with a set of tree operations, instead of a single graph, provides more flexibility and opens the way for modeling such important concepts as inheritance - to mention just one. To the best of our knowledge, no other model uses the concept of tree as the basic database structure at the conceptual level.\smallskip 

\noindent Our model considers functions as `first class citizens' forming trees, while relations are derived concepts. This is in sharp contrast with the relational model, in which it is the relations that are considered as `first class citizens' and functional dependencies are simply `constraints' that the relations must satisfy. Put differently, in the relational model, functional dependencies play a static role, whereas in our model they play a dynamic role in the sense that they express key-value properties between nodes to represent information. In this sense, our model is closer to NoSQL databases\cite{DBLP:journals/cacm/Stonebraker10}  and in particular to key-value databases   \cite{DBLP:conf/data/FernandesB18}. Indeed, each edge $f: X\to Y$ of a tree in our model represents a total function from node $X$ to node $Y$; and as such, each edge $f$ can be seen as a set of key-value pairs $(x, f(x))$ where $x$ is the key and $f(x)$ is the value. Therefore, in principle, it should be possible to store the key-value pairs of a tree database using a key-value store. We envisage this as a long term perspective of our present work.\smallskip


\noindent Our current work develops along a number of research lines that we describe succinctly below: \smallskip

\noindent {\bf \em Query rewriting}: Since our model is purely functional, the functional algebra provides the appropriate framework for studying query rewriting. We have given two examples of such rewriting, namely traversal query rewriting based on the fact that composition distributes over pairing (Section \ref{sec:TQ}) and analytic query rewriting based on the composition rule (Section \ref{RAE}). However, more work is needed in order to develop a complete rewriting system for optimizing the evaluation of traversal and analytic queries. \smallskip
   
\noindent {\bf \em Updating}: Adding or removing a key-value pair $(x, f(x))$ in an edge $f$ while preserving function totality may provoke updates in other edges of a tree (and eventually in all edges). 
Hence the need for efficient algorithms for the propagation of updates throughout the database. Additionally, such algorithms will have to take into account the eventual presence of the two types of integrity constraints that we have seen, namely refinement constraints and set theoretic constraints (see Section \ref{IC}) - as well as constraints due to the presence of $Isa$ edges in the database. Designing algorithms that allow updating a tree database while satisfying all database constraints is an important theme that needs further work. 

\smallskip
\noindent {\bf \em Inheritance}: Our intention in this paper was to simply show that our model can support classical inheritance in a seamless manner. We did so by considering trees that are disjoint. However, the concept of inheritance can be generalized to trees that are not disjoint. Indeed, even if two trees $\mathcal T$ and $\mathcal T'$ share one or more nodes they might have subtrees which are disjoint. For example, a subtree of $\mathcal T$ might be disjoint with $\mathcal T'$ (or vice versa); or there might be pairs of subtrees (one in each tree) that are disjoint. Such subtrees can be obtained by tree projections and/or star tree partition (see Section \ref{subsec:TP}). So the problem here can be stated as follows: given a set of trees (e.g. those in a tree database) use tree projection and/or tree partition to find a (new) set of trees of interest connected through $Isa$ edges and offering to the user a `web of trees' in which each $Isa$ edge from a tree $\tau$ to a tree $\tau'$ `imports' or `adds' in $\tau$ the information contained in $\tau'$. In such a web of trees the $Isa$ edges act as `hyperlinks' similar to those in the World Wide Web. 

\noindent Our current work aims at defining such webs of trees having specified properties and serving specific application environments (e.g. organized online collections of documents, books, audio, video, and other multimedia files that users can access electronically). \smallskip

\noindent {\bf \em Tree database design} We have seen that a tree database is a set $\mathcal D$ of trees representing the datasets of an application and their relationships. Clearly there might be several choices of trees to put in the database following a set of criteria (e.g. small trees with an easy to grasp structure versus larger trees expressing more closely the structure of the application). Hence the question: given a set  $\mathcal D$ of trees representing an application how can we find an  `equivalent' set  $\mathcal D'$  of trees satisfying a set of criteria? To this end one would obviously need to use the tree algebra which allows to decompose large trees into smaller trees or to compose small trees to create larger ones.  So the relevant question here is: what set of criteria should one adopt. Two examples of such criteria from relational databases are: (a) preservation of information and (b) preservation of dependencies \cite{Ullman83}. 

\noindent Designing algorithms that take the set  $\mathcal D$ as input and produce as output an equivalent set $\mathcal D'$  satisfying a set of given criteria is what we mean by `tree database design' and this is a main topic of our current research. \smallskip

\noindent {\bf \em Revisiting relational model theory}: In the previous section, we have seen how a consistent relational database can be defined on top of a tree database. We have also seen how we can use the functional algebra to `simulate' Armstrong's axioms (Section \ref{subsec:FA}, Lemma \ref{AA}). In view of these results it is natural to revisit relational database theory and get new insights into functional dependency theory as well as into schema design theory for relational databases. \\

\noindent {\bf \em Data visualization}: We are currently investigating the development of a visualization model, proper to tree databases, following the concepts introduced in previous work \cite{SpyratosS19}. Such a model should allow the visualization and visual exploration of query results. \\

\noindent As a longer term project, we envisage the design and implementation of a full-fledged tree-database management system in which storing, accessing and analyzing data is based on the concept of tree as the basic data structure. Such a system should also provide support for inheritance through tree composition, for interfacing relational data warehouses and for the visualization and visual exploration of query results.

\section*{Acknowledgements}
The author would like to thank Professor Dominique Laurent for his comments that helped improve the content of this paper.


\bibliographystyle{plain}
\bibliography{biblio}

@inproceedings{LaurentS11,
  author       = {Dominique Laurent and
                  Nicolas Spyratos},
  editor       = {William I. Grosky and
                  Youakim Badr and
                  Richard Chbeir},
  title        = {Rewriting aggregate queries using functional dependencies},
  booktitle    = {{MEDES} '11: International {ACM} Conference on Management of Emergent
                  Digital EcoSystems, San Francisco, CA, USA, November 21-24, 2011},
  pages        = {40--47},
  publisher    = {{ACM}},
  year         = {2011},
  url          = {https://doi.org/10.1145/2077489.2077497},
  doi          = {10.1145/2077489.2077497},
  bibsource    = {dblp computer science bibliography, https://dblp.org}
}

@inproceedings{Ullman83,
  author    = {Jeffrey D. Ullman},
  editor    = {R. E. A. Mason},
  title     = {Universal Relation Interfaces for Database Systems},
  booktitle = {Information Processing 83, Proceedings of the {IFIP} 9th World Computer
               Congress, Paris, France, September 19-23, 1983},
  pages     = {243--252},
  publisher = {North-Holland/IFIP},
  year      = {1983},
  bibsource = {dblp computer science bibliography, https://dblp.org}
}

@BOOK{Ullman,
  author =       {Jeffrey D. Ullman},
  title =        {Principles of Databases and Knowledge-Base Systems},
  publisher =    {Computer Science Press},
  year =         {1988},
  volume =       {1-2},
}

@INPROCEEDINGS{BR,
  AUTHOR =       "K.S. Beyer and R. Ramakrishnan",
  TITLE =        "Bottom-Up Computation of Sparse and Iceberg CUBEs",
  BOOKTITLE =    "ACM SIGMOD",
  YEAR =         "1999",
  pages =        "359-370",
  address =      "Philadelphia, Pennysylvania, USA",
}

@article{SpyratosS18,
  author    = {Nicolas Spyratos and
               Tsuyoshi Sugibuchi},
  title     = {{HIFUN} - a high level functional query language for big data analytics},
  journal   = {J. Intell. Inf. Syst.},
  volume    = {51},
  number    = {3},
  pages     = {529--555},
  year      = {2018},
  url       = {https://doi.org/10.1007/s10844-018-0495-6},
  doi       = {10.1007/s10844-018-0495-6},
  bibsource = {dblp computer science bibliography, https://dblp.org}
}

@inproceedings{SpyratosS19,
  author    = {Nicolas Spyratos and
               Tsuyoshi Sugibuchi},
  editor    = {Alfredo Cuzzocrea and
               Sergio Greco and
               Henrik Legind Larsen and
               Domenico Sacc{\`{a}} and
               Troels Andreasen and
               Henning Christiansen},
  title     = {Data Exploration in the {HIFUN} Language},
  booktitle = {Flexible Query Answering Systems - 13th International Conference,
               {FQAS} 2019, Amantea, Italy, July 2-5, 2019, Proceedings},
  series    = {Lecture Notes in Computer Science},
  volume    = {11529},
  pages     = {176--187},
  publisher = {Springer},
  year      = {2019},
  url       = {https://doi.org/10.1007/978-3-030-27629-4\_18},
  doi       = {10.1007/978-3-030-27629-4\_18},
  bibsource = {dblp computer science bibliography, https://dblp.org}
}

@article{PapadakiST21,
  author    = {Maria{-}Evangelia Papadaki and
               Nicolas Spyratos and
               Yannis Tzitzikas},
  title     = {Towards Interactive Analytics over {RDF} Graphs},
  journal   = {Algorithms},
  volume    = {14},
  number    = {2},
  pages     = {34},
  year      = {2021},
  url       = {https://doi.org/10.3390/a14020034},
  doi       = {10.3390/a14020034},
  bibsource = {dblp computer science bibliography, https://dblp.org}
}

@article{ZervoudakisKSP21,
  author    = {Petros Zervoudakis and
               Haridimos Kondylakis and
               Nicolas Spyratos and
               Dimitris Plexousakis},
  title     = {Query Rewriting for Incremental Continuous Query Evaluation in {HIFUN}},
  journal   = {Algorithms},
  volume    = {14},
  number    = {5},
  pages     = {149},
  year      = {2021},
  url       = {https://doi.org/10.3390/a14050149},
  doi       = {10.3390/a14050149},
  bibsource = {dblp computer science bibliography, https://dblp.org}
}

@article{VitsaxakiNMS22,
author = {Katerina Vitsaxaki and Stavroula Ntoa and George Margetis and Nicolas Spyratos},
title = {Interactive Visual Exploration of Big Relational Datasets},
journal = {International Journal of Human–Computer Interaction}, 
year = {2022},
url = {https://doi.org/10.1080/10447318.2022.2073007},
doi = {10.1080/10447318.2022.2073007}
}

@inproceedings{DBLP:conf/sigmod/ArenasGS21,
  author       = {Marcelo Arenas and
                  Claudio Gutierrez and
                  Juan F. Sequeda},
  editor       = {Guoliang Li and
                  Zhanhuai Li and
                  Stratos Idreos and
                  Divesh Srivastava},
  title        = {Querying in the Age of Graph Databases and Knowledge Graphs},
  booktitle    = {{SIGMOD} '21: International Conference on Management of Data, Virtual
                  Event, China, June 20-25, 2021},
  pages        = {2821--2828},
  publisher    = {{ACM}},
  year         = {2021},
  url          = {https://doi.org/10.1145/3448016.3457545},
  doi          = {10.1145/3448016.3457545},
  bibsource    = {dblp computer science bibliography, https://dblp.org}
}

@inproceedings{DBLP:conf/cisim/Pokorny15,
  author       = {Jaroslav Pokorn{\'{y}}},
  editor       = {Khalid Saeed and
                  Wladyslaw Homenda},
  title        = {Graph Databases: Their Power and Limitations},
  booktitle    = {Computer Information Systems and Industrial Management - 14th {IFIP}
                  {TC} 8 International Conference, {CISIM} 2015, Warsaw, Poland, September
                  24-26, 2015. Proceedings},
  series       = {Lecture Notes in Computer Science},
  volume       = {9339},
  pages        = {58--69},
  publisher    = {Springer},
  year         = {2015},
  url          = {https://doi.org/10.1007/978-3-319-24369-6\_5},
  doi          = {10.1007/978-3-319-24369-6\_5},
  bibsource    = {dblp computer science bibliography, https://dblp.org}
}

@incollection{DBLP:reference/bdt/GutierrezHW19,
  author       = {Claudio Gutierrez and
                  Jan Hidders and
                  Peter T. Wood},
  editor       = {Sherif Sakr and
                  Albert Y. Zomaya},
  title        = {Graph Data Models},
  booktitle    = {Encyclopedia of Big Data Technologies},
  publisher    = {Springer},
  year         = {2019},
  url          = {https://doi.org/10.1007/978-3-319-63962-8\_81-1},
  doi          = {10.1007/978-3-319-63962-8\_81-1},
  bibsource    = {dblp computer science bibliography, https://dblp.org}
}

@article{DBLP:journals/corr/AnglesABHRV16,
  author       = {Renzo Angles and
                  Marcelo Arenas and
                  Pablo Barcel{\'{o}} and
                  Aidan Hogan and
                  Juan L. Reutter and
                  Domagoj Vrgoc},
  title        = {Foundations of Modern Graph Query Languages},
  journal      = {CoRR},
  volume       = {abs/1610.06264},
  year         = {2016},
  url          = {http://arxiv.org/abs/1610.06264},
  eprinttype    = {arXiv},
  eprint       = {1610.06264},
  bibsource    = {dblp computer science bibliography, https://dblp.org}
}

@inproceedings{DBLP:conf/aib/Hogan22,
  author       = {Aidan Hogan},
  editor       = {Camille Bourgaux and
                  Ana Ozaki and
                  Rafael Pe{\~{n}}aloza},
  title        = {Knowledge Graphs: {A} Guided Tour (Invited Paper)},
  booktitle    = {International Research School in Artificial Intelligence in Bergen,
                  {AIB} 2022, June 7-11, 2022, University of Bergen, Norway},
  series       = {OASIcs},
  volume       = {99},
  pages        = {1:1--1:21},
  publisher    = {Schloss Dagstuhl - Leibniz-Zentrum f{\"{u}}r Informatik},
  year         = {2022},
  url          = {https://doi.org/10.4230/OASIcs.AIB.2022.1},
  doi          = {10.4230/OASIcs.AIB.2022.1},
  bibsource    = {dblp computer science bibliography, https://dblp.org}
}

@article{DBLP:journals/csur/HoganBCdMGKGNNN21,
  author       = {Aidan Hogan and
                  Eva Blomqvist and
                  Michael Cochez and
                  Claudia d'Amato and
                  Gerard de Melo and
                  Claudio Gutierrez and
                  Sabrina Kirrane and
                  Jos{\'{e}} Emilio Labra Gayo and
                  Roberto Navigli and
                  Sebastian Neumaier and
                  Axel{-}Cyrille Ngonga Ngomo and
                  Axel Polleres and
                  Sabbir M. Rashid and
                  Anisa Rula and
                  Lukas Schmelzeisen and
                  Juan F. Sequeda and
                  Steffen Staab and
                  Antoine Zimmermann},
  title        = {Knowledge Graphs},
  journal      = {{ACM} Comput. Surv.},
  volume       = {54},
  number       = {4},
  pages        = {71:1--71:37},
  year         = {2022},
  url          = {https://doi.org/10.1145/3447772},
  doi          = {10.1145/3447772},
  bibsource    = {dblp computer science bibliography, https://dblp.org}
}

@article{DBLP:journals/csur/Wiederhold95,
  author       = {Gio Wiederhold},
  title        = {Mediation in Information Systems},
  journal      = {{ACM} Comput. Surv.},
  volume       = {27},
  number       = {2},
  pages        = {265--267},
  year         = {1995},
  url          = {https://doi.org/10.1145/210376.210390},
  doi          = {10.1145/210376.210390},
  bibsource    = {dblp computer science bibliography, https://dblp.org}
}

@article{spyratos2023context,
  title={The Context Model: A Graph Database Model},
  author={Nicolas Spyratos},
  year={2023},
  doi = {10.48550/ARXIV.2305.13895},  
  url = {https://arxiv.org/abs/2305.13895},  
  journal   = {CoRR},
  volume    = {abs/2305.13895},
  publisher = {arXiv},  
  copyright = {Creative Commons Attribution 4.0 International}
}

@incollection{DBLP:reference/db/Pitoura18d,
  author       = {Evaggelia Pitoura},
  editor       = {Ling Liu and
                  M. Tamer {\"{O}}zsu},
  title        = {Query Optimization},
  booktitle    = {Encyclopedia of Database Systems, Second Edition},
  publisher    = {Springer},
  year         = {2018},
  url          = {https://doi.org/10.1007/978-1-4614-8265-9\_861},
  doi          = {10.1007/978-1-4614-8265-9\_861},
  bibsource    = {dblp computer science bibliography, https://dblp.org}
}

@book{DBLP:books/cs/Maier83,
  author       = {David Maier},
  title        = {The Theory of Relational Databases},
  publisher    = {Computer Science Press},
  year         = {1983},
  url          = {http://web.cecs.pdx.edu/\%7Emaier/TheoryBook/TRD.html},
  isbn         = {0-914894-42-0},
  bibsource    = {dblp computer science bibliography, https://dblp.org}
}

@inproceedings{DBLP:conf/ifip/Armstrong74,
  author       = {William Ward Armstrong},
  editor       = {Jack L. Rosenfeld},
  title        = {Dependency Structures of Data Base Relationships},
  booktitle    = {Information Processing, Proceedings of the 6th {IFIP} Congress 1974,
                  Stockholm, Sweden, August 5-10, 1974},
  pages        = {580--583},
  publisher    = {North-Holland},
  year         = {1974},
  bibsource    = {dblp computer science bibliography, https://dblp.org}
}

@article{0c1f962ff899492c9ea15bd38a0ffb65,
title = "A survey of typical attributed graph queries",
author = "Yanhao Wang and Yuchen Li and Ju Fan and Chang Ye and Mingke Chai",
year = "2021",
month = jan,
doi = "10.1007/s11280-020-00849-0",
language = "English",
volume = "24",
pages = "297--346",
journal = "World Wide Web (New York)",
issn = "1386-145X",
publisher = "SPRINGER NEW YORK LLC",
}

@article{DBLP:journals/csur/Taivalsaari96,
  author       = {Antero Taivalsaari},
  title        = {On the Notion of Inheritance},
  journal      = {{ACM} Comput. Surv.},
  volume       = {28},
  number       = {3},
  pages        = {438--479},
  year         = {1996},
  url          = {https://doi.org/10.1145/243439.243441},
  doi          = {10.1145/243439.243441},
  bibsource    = {dblp computer science bibliography, https://dblp.org}
}

@misc{anuyah2024understandinggraphdatabasescomprehensive,
      title={Understanding Graph Databases: A Comprehensive Tutorial and Survey}, 
      author={Sydney Anuyah and Victor Bolade and Oluwatosin Agbaakin},
      year={2024},
      eprint={2411.09999},
      archivePrefix={arXiv},
      primaryClass={cs.DB},
      url={https://arxiv.org/abs/2411.09999}, 
}

@article{DBLP:journals/csys/Waldo91,
  author       = {Jim Waldo},
  title        = {Controversy: The Case for Multiple Inheritance in {C++}},
  journal      = {Comput. Syst.},
  volume       = {4},
  number       = {2},
  pages        = {157--171},
  year         = {1991},
  url          = {http://www.usenix.org/publications/compsystems/1991/spr\_waldo.pdf},
  bibsource    = {dblp computer science bibliography, https://dblp.org}
}

@inproceedings{DBLP:conf/data/FernandesB18,
  author       = {Diogo Fernandes and
                  Jorge Bernardino},
  editor       = {Jorge Bernardino and
                  Christoph Quix},
  title        = {Graph Databases Comparison: AllegroGraph, ArangoDB, InfiniteGraph,
                  Neo4J, and OrientDB},
  booktitle    = {Proceedings of the 7th International Conference on Data Science, Technology
                  and Applications, {DATA} 2018, Porto, Portugal, July 26-28, 2018},
  pages        = {373--380},
  publisher    = {SciTePress},
  year         = {2018},
  url          = {https://doi.org/10.5220/0006910203730380},
  doi          = {10.5220/0006910203730380},
  bibsource    = {dblp computer science bibliography, https://dblp.org}
}

@article{DBLP:journals/cacm/Stonebraker10,
  author       = {Michael Stonebraker},
  title        = {{SQL} databases v. NoSQL databases},
  journal      = {Commun. {ACM}},
  volume       = {53},
  number       = {4},
  pages        = {10--11},
  year         = {2010},
  url          = {https://doi.org/10.1145/1721654.1721659},
  doi          = {10.1145/1721654.1721659},
  bibsource    = {dblp computer science bibliography, https://dblp.org}
}

@misc{Mendelzon,
  author = {Alberto O. Mendelzon},
  title = {Who won the Universal Relation Wars?},
  year = {2002},
  note = {Symposium in honor of Jeffrey D. Ullman. Stanford University~--~Last accessed 4 March 2026~--~Available at {\tt http://infolab.stanford.edu/jdu-symposium/talks/mendelzon.pdf}},
  url = {http://infolab.stanford.edu/jdu-symposium/talks/mendelzon.pdf}
}

\end{document}